\documentclass[UKenglish,a4paper,pdfa,cleveref]{lipics-v2021}

\usepackage{booktabs,csquotes,etoolbox,mathtools,tkz-euclide,xparse}
\usepackage[ruled]{algorithm}
\usepackage[noend]{algpseudocode}

\MakeRobust{\Call}
\makeatletter
\patchcmd{\ALG@doentity}{\item[]\nointerlineskip}{}{}{}
\makeatother

\title{Uncertain Curve Simplification}
\author{Kevin Buchin}{Department of Mathematics and Computer Science, TU Eindhoven, Netherlands\and\url{https://www.win.tue.nl/~kbuchin/}}{k.a.buchin@tue.nl}{https://orcid.org/0000-0002-3022-7877}{}
\author{Maarten L\"offler}{Department of Information and Computing Sciences, Utrecht University, Netherlands\and\url{https://webspace.science.uu.nl/~loffl001/}}{m.loffler@uu.nl}{}{Partially supported by the Dutch Research Council (NWO) under project no.\@ 614.001.504.}
\author{Aleksandr Popov}{Department of Mathematics and Computer Science, TU Eindhoven, Netherlands\and\url{https://www.win.tue.nl/~apopov/}}{a.popov@tue.nl}{https://orcid.org/0000-0002-0158-1746}{Supported by the Dutch Research Council (NWO) under project no.\@ 612.001.801.}
\author{Marcel Roeloffzen}{Department of Mathematics and Computer Science, TU Eindhoven, Netherlands\and\url{https://www.win.tue.nl/~mroeloff/}}{m.j.m.roeloffzen@tue.nl}{}{Supported by the Dutch Research Council (NWO) under project no.\@ 628.011.005.}
\authorrunning{K.~Buchin, M.~L\"offler, A.~Popov, and M.~Roeloffzen}
\Copyright{Kevin Buchin, Maarten L\"offler, Aleksandr Popov, and Marcel Roeloffzen}

\nolinenumbers
\hideLIPIcs

\theoremstyle{definition}
\newtheorem{problem}[theorem]{Problem}

\newcommand*{\eqdef}{\mathrel{\overset{\makebox[0pt]{\(\mathrm{\scriptscriptstyle def}\)}}{=}}}
\newcommand*{\reason}[1]{\mathrel{\hphantom{=}}\text{\{#1\}}}
\newcommand*{\fr}{d_\mathrm{F}}
\newcommand*{\hs}{d_\mathrm{H}}
\newcommand*{\R}{\mathbb{R}}
\newcommand*{\N}{\mathbb{N}}
\newcommand*{\Rt}{\R^2}
\newcommand*{\U}{\mathcal{U}}
\newcommand*{\Pm}{\mathcal{P}}
\newcommand*{\bO}{\mathcal{O}}
\newcommand*{\True}{\textnormal{True}}
\newcommand*{\False}{\textnormal{False}}
\DeclareMathOperator*{\argmin}{argmin}
\DeclareMathOperator*{\argmax}{argmax}

\ExplSyntaxOn
\NewDocumentCommand{\multiadjustlimits}{m}{
    \group_begin:
    \multiadjustlimits_measure:n { #1 }
    \multiadjustlimits_print:n { #1 }
    \group_end:
}
\tl_new:N  \l__multiadjustlimits_operator_tl
\tl_new:N  \l__multiadjustlimits_limit_tl
\cs_new_protected:Nn \multiadjustlimits_measure:n{
    \clist_map_function:nN { #1 } \__multiadjustlimits_measure:n
}
\cs_new_protected:Nn \__multiadjustlimits_measure:n{
    \__multiadjustlimits_measure:NNn #1
}
\cs_new_protected:Nn \__multiadjustlimits_measure:NNn{
    \tl_put_right:Nn \l__multiadjustlimits_operator_tl { #1 }
    \tl_put_right:Nn \l__multiadjustlimits_limit_tl { #3 }
}
\cs_new_protected:Nn \multiadjustlimits_print:n{
    \clist_map_function:nN { #1 } \__multiadjustlimits_print:n
}
\cs_new_protected:Nn \__multiadjustlimits_print:n{
    \__multiadjustlimits_print:NNn #1
}
\cs_new_protected:Nn \__multiadjustlimits_print:NNn{
    \mathop { \vphantom{\l__multiadjustlimits_operator_tl} \mathopen{} #1 }
    \limits
    \sb{ \vphantom{\l__multiadjustlimits_limit_tl} #3 }
}
\ExplSyntaxOff

\begin{document}
\maketitle

\begin{abstract}
We study the problem of polygonal curve simplification under uncertainty, where instead of a
sequence of exact points, each uncertain point is represented by a region, which contains the
(unknown) true location of the vertex.
The regions we consider are disks, line segments, convex polygons, and discrete sets of points.
We are interested in finding the shortest subsequence of uncertain points such that no matter what
the true location of each uncertain point is, the resulting polygonal curve is a valid
simplification of the original polygonal curve under the Hausdorff or the Fr\'echet distance.
For both these distance measures, we present polynomial-time algorithms for this problem.
\end{abstract}

\section{Introduction}
In this paper, we investigate the topic of curve simplification under uncertainty.
There are many classical algorithms dealing with curve simplification with different distance
metrics; however, it is typically assumed that the locations of points making up the curves are
known precisely, which is often not ideal when modelling real-life data.
An obvious example highlighting the necessity of taking uncertainty into account comes with GPS
data, where each measured location comes with inherent uncertainty due to the physical
characteristics of the measurement.
Curve simplification is often used as a first step to reduce the noise-to-signal ratio in the
trajectory data before applying other algorithms or when storing large amounts of data.
In both cases modelling uncertainty could reduce the error introduced by simplifying imprecise
measurements while maintaining a short, efficient representation of the data.

There is a large volume of foundational previous work in the area of curve
simplification~\cite{agarwal:2000}, including work on vertex-constrained
simplification, such as the well-known algorithms by Ramer and by Douglas and
Peucker~\cite{douglas_peucker,ramer} using the Hausdorff distance, by Imai and Iri~\cite{imai_iri}
using either the Hausdorff or the Fr\'echet distance, by Agarwal et al.~\cite{agarwal} using the
Fr\'echet distance, and various improvements and related
approaches~\cite{barequet,bringmann_simpl,buchin_progressive,chan_chin,gudmundsson,guibas,melkman,kerkhof}.
In particular, the basic approach of the Imai--Iri algorithm involves computing the \emph{shortcut
graph,} which captures all the possible simplifications of a curve, and then finding a shortest path
through the graph in terms of the number of edges from the start node to the end node, thus finding
the simplification with fewest edges.
We adapt this approach in our work to the setting with uncertainty.

There have recently been some advances in the study of uncertainty in computational geometry,
including work on maximising and minimising various measures on uncertain
points~\cite{jorgensen:2011,knauer_hausdorff,loeffler,prob_loeffler,loeffler:2006},
triangulations~\cite{buchin:2011,loeffler:2010,kreveld:2010}, visibility in uncertain
polygons~\cite{visibility}, moving points~\cite{minimising_ply}, and other
problems~\cite{agarwal:2016,agarwal:2017,driemel:2013,gray:2012,loeffler:2014,pei:2007,suri:2013}.
There is also work by Ahn et al.~\cite{ahn_imprecise}, and, more recently, by Buchin et
al.~\cite{uncertaincurves,popov} on various minimisation and maximisation variants involving curve
similarity with the Fr\'echet distance under uncertainty, as well as other work in combinations of
curve analysis and uncertainty~\cite{bbmm_segment,bbmm_patterns,sijben_dfd}.
To our knowledge, there is no previous work tackling curve simplification under uncertainty.

\begin{figure}
\centering
\includegraphics{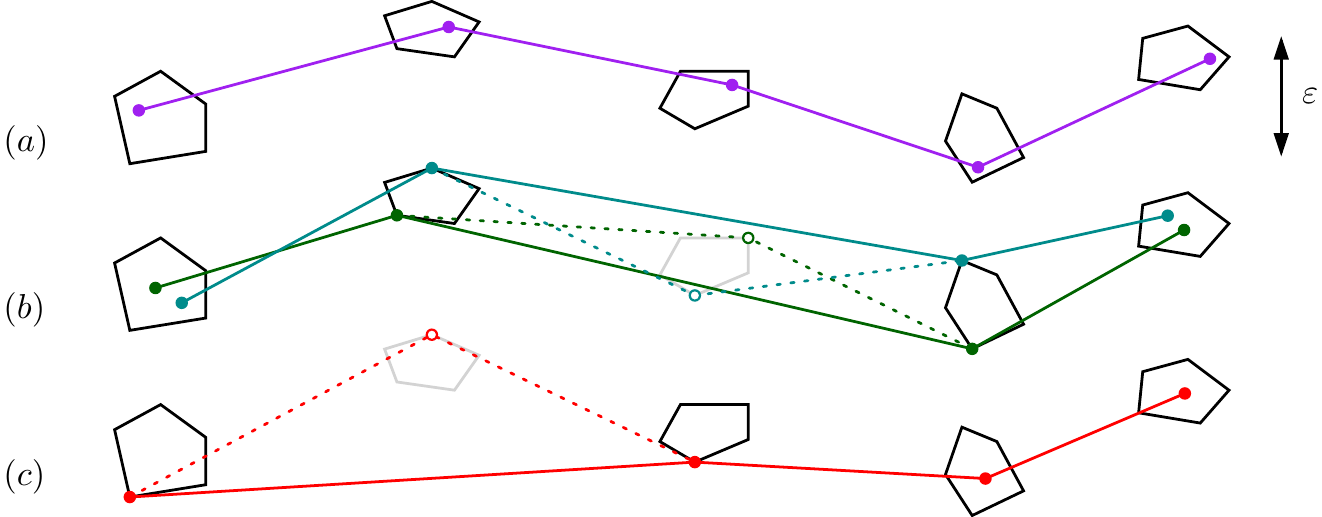}
\caption{(\textit{a}) An uncertain curve modelled with convex polygons and a potential realisation.
(\textit{b}) A valid simplification under the Hausdorff distance with the threshold \(\varepsilon\):
for every realisation, the subsequence is within Hausdorff distance \(\varepsilon\) from the full
sequence.
(\textit{c}) An invalid simplification under the Hausdorff distance with the threshold
\(\varepsilon\): there is a realisation for which the subsequence is not within Hausdorff distance
\(\varepsilon\) from the full sequence.}
\label{fig:example}
\end{figure}

In this paper, we adopt the locational model for the uncertain points: we know that each point
exists, but we do not know its precise location.
It can be represented as a discrete set of points, of which only one is the true location; we say
that this model uses \emph{indecisive} points.
We also use \emph{imprecise} points, modelled as a compact continuous set, such as disks, line
segments, or convex polygons; again, the true location is one (unknown) point from the set.
An \emph{uncertain curve} is a sequence of uncertain points of the same kind.
A \emph{realisation} of an uncertain curve is a precise polygonal curve obtained by taking one point
from each uncertain point.

In this paper, we solve the following problem, illustrated in \cref{fig:example}: \emph{given an
uncertain curve as a sequence of \(n\) uncertain points, find the shortest subsequence of the
uncertain points of the curve such that for any realisation of the curve, the corresponding
realisation of the subsequence is a valid simplification of that realisation.}
We present a family of efficient algorithms for this problem under both the Hausdorff and the
Fr\'echet distance, with the uncertain points modelled as indecisive points, as disks, as line
segments, and as convex polygons, shown in \cref{tab:runningtime}.

\begin{table}
\centering
\caption{Running time of our approach in each setting.
For indecisive points, \(k\) is the number of options per point.
For convex polygons, \(k\) is the number of vertices.}
\begin{tabular}{l c c c c}
\toprule
                   & Indecisive      & Disks        & Line segments & Convex polygons\\
\midrule
Hausdorff distance & \(\bO(n^3k^3)\) & \(\bO(n^3)\) & \(\bO(n^3)\)  & \(\bO(n^3k^3)\)\\
Fr\'echet distance & \(\bO(n^3k^3)\) & \(\bO(n^3)\) & \(\bO(n^3)\)  & \(\bO(n^3k^3)\)\\
\bottomrule
\end{tabular}
\label{tab:runningtime}
\end{table}

\section{Preliminaries}\label{sec:prelims}
Denote\footnote{We use \(\coloneqq\) and \(\eqqcolon\) to denote assignment, \(\eqdef\) for
equivalent quantities in definitions or to point out equality by earlier definition, and \(=\) in
other contexts.
We also use \(\equiv\), but its usage is always explained.}
\([n] \eqdef \{1, 2, \dots, n\}\) for any \(n \in \N^{> 0}\).
Given two points \(p, q \in \Rt\), denote their Euclidean distance with \(\lVert p - q\rVert\).

Denote a \emph{sequence} of points in \(\Rt\) with \(\pi = \langle p_1, \dots, p_n\rangle\).
For only two points \(p, q \in \Rt\), we also use \(pq\) instead of \(\langle p, q\rangle\).
Denote a subsequence of a sequence \(\pi\) from index \(i\) to \(j\) with
\(\pi[i: j] = \langle p_i, p_{i + 1}, \dots, p_j\rangle\).
This notation can also be applied if we interpret \(\pi\) as a \emph{polygonal curve} on \(n\)
vertices (of \emph{length} \(n\)).
It is defined by linearly interpolating between the successive points in the sequence and can be
seen as a continuous function, for \(i \in [n - 1]\) and \(\alpha \in [0, 1]\):
\[\pi(i + \alpha) = (1 - \alpha) p_i + \alpha p_{i + 1}\,.\]

We also introduce the notation for the order of points along a curve.
Let \(p \coloneqq \pi(a)\) and \(q \coloneqq \pi(b)\) for some \(a, b \in [1, n]\).
Then \(p \prec q\) iff \(a < b\), \(p \preccurlyeq q\) iff \(a \leq b\), and \(p \equiv q\) iff
\(a = b\).
Note that we can have \(p = q\) for \(a \neq b\) if the curve intersects itself.

Finally, given points \(p, q, r \in \Rt\), define the distance from \(p\) to the segment \(qr\) as
\[d(p, qr) \eqdef \min_{t \in qr} \lVert p - t\rVert\,.\]

An \emph{uncertainty region} \(U \subset \Rt\) describes a possible location of a true point: it has
to be inside the region, but there is no information as to where exactly.
We use several uncertainty models, so the regions \(U\) are of different shape.
An \emph{indecisive point} is a form of an uncertain point where the uncertainty region is
represented as a discrete set of points, and the true point is one of them:
\(U = \{p^1, \dots, p^k\}\), with \(k \in \N^{> 0}\) and \(p^i \in \Rt\) for all \(i \in [k]\).
\emph{Imprecise points} are modelled with uncertainty regions that are compact continuous sets.
In particular, we consider \emph{disks} and \emph{polygonal closed convex sets.}
We denote a disk with the centre \(c \in \Rt\) and the radius \(r \in \R^{\geq 0}\) as \(D(c, r)\).
Formally, \(D(c, r) \eqdef \{p \in \Rt \mid \lVert p - c \rVert \leq r\}\).
Define a \emph{polygonal closed convex set (PCCS)} as a closed convex set with bounded area that can
be described as the intersection of a \emph{finite} number of closed half-spaces.
Note that this definition includes both convex polygons and line segments (in~2D).
Given a PCCS \(U\), let \(V(U)\) denote the set of vertices of \(U\), i.e.\@ vertices of a convex
polygon or endpoints of a line segment.

We call a sequence of uncertainty regions an \emph{uncertain curve:}
\(\U = \langle U_1, \dots, U_n\rangle\).
If we pick a point from each uncertainty region of \(\U\), we get a polygonal curve \(\pi\) that we
call a \emph{realisation} of \(\U\) and denote it with \(\pi \Subset \U\).
That is, if for some \(n \in \N^{> 0}\) we have \(\pi = \langle p_1, \dots, p_n\rangle\) and
\(\U = \langle U_1, \dots, U_n\rangle\), then \(\pi \Subset \U\) if and only if \(p_i \in U_i\) for
all \(i \in [n]\).

Suppose we are given a polygonal curve \(\pi = \langle p_1, \dots, p_n\rangle\), a threshold
\(\varepsilon \in \R^{\geq 0}\), and a curve built on the subsequence of vertices of \(\pi\) for
some set \(I = \{i_1, \dots, i_\ell\} \subseteq [n]\), i.e.\@
\(\sigma = \langle p_{i_1}, \dots, p_{i_\ell}\rangle\) with \(i_j < i_{j + 1}\) for all
\(j \in [\ell - 1]\) and \(\ell \leq n\).
We call \(\sigma\) an \emph{\(\varepsilon\)-simplification} of \(\pi\) if for each segment
\(\langle p_{i_j}, p_{i_{j + 1}}\rangle\), we have
\(\delta(\langle p_{i_j}, p_{i_{j + 1}}\rangle, \pi[i_j: i_{j + 1}]) \leq \varepsilon\),
where \(\delta\) denotes some distance measure, e.g.\@ the Hausdorff or the Fr\'echet distance.

The \emph{Hausdorff distance} between two sets \(P, Q \subset \Rt\) is defined as
\[\hs(P, Q) \eqdef \max\big\{\adjustlimits\sup_{p \in P} \inf_{q \in Q} \lVert p - q\rVert,
\adjustlimits\sup_{q \in Q} \inf_{p \in P} \lVert p - q\rVert\big\}\,.\]
For two polygonal curves \(\pi\) and \(\sigma\) in \(\Rt\), since \(\pi\) and \(\sigma\) are closed
and bounded, we get
\[\hs(\pi, \sigma) = \max\big\{\adjustlimits\max_{p \in \pi} \min_{q \in \sigma} \lVert p - q\rVert,
\adjustlimits\max_{q \in \sigma} \min_{p \in \pi} \lVert p - q\rVert\big\}\,.\]

\begin{figure}[tb]
\begin{minipage}{.48\linewidth}
\begin{tikzpicture}[scale=.76]
\pgfmathsetmacro{\thr}{2}
\tkzDefPoint(0,0){p_1}
\tkzDefPoint(2,2){p_2}
\tkzDefPoint(3,-2){p_3}
\tkzDefPoint(7,1){p_4}
\tkzDefPoint(5,-1){p_5}
\tkzDefPoint(8,0){p_6}

\tkzDefPoint(2,0){s_2}
\tkzDefPoint(3,0){s_3}
\tkzDefPoint(7,0){s_4}
\tkzDefPoint(5,0){s_5}

\tkzDefPoint(3.5,1.5){a}
\tkzDefPoint(5.5,1.5){b}
\tkzDefPoint(4.5,1.5){e}

\tkzDrawSegments(p_1,p_2 p_2,p_3 p_3,p_4 p_4,p_5 p_5,p_6)
\tkzDrawSegments[thick,gray](p_1,p_6)
\tkzDrawSegments[thick,dotted](p_2,s_2 p_3,s_3 p_4,s_4 p_5,s_5 a,b)
\tkzDrawPoints(p_1,p_2,p_3,p_4,p_5,p_6,s_2,s_3,s_4,s_5)
\tkzLabelPoint[above](e){\(\varepsilon\)}
\tkzLabelPoints(p_1,p_3,p_5,p_6,s_3)
\tkzLabelPoints[above right](p_4,s_4)
\tkzLabelPoints[right](p_2)
\tkzLabelPoints[above](s_5)
\tkzLabelPoints[below](s_2)
\end{tikzpicture}
\end{minipage}\hfill%
\begin{minipage}{.48\linewidth}
\begin{tikzpicture}[scale=.76]
\pgfmathsetmacro{\thr}{2}
\tkzDefPoint(0,0){p_1}
\tkzDefPoint(2,2){p_2}
\tkzDefPoint(3,-2){p_3}
\tkzDefPoint(7,1){p_4}
\tkzDefPoint(5,-1){p_5}
\tkzDefPoint(8,0){p_6}

\tkzDefPoint(2,0){s_2}
\tkzDefPoint(3,0){s_3}
\tkzInterLC[R](p_1,p_6)(p_4,\thr cm)
\tkzGetSecondPoint{s_4}
\tkzDefShiftPoint[s_4](0,0){s_5}

\tkzDefPoint(3.5,1.5){a}
\tkzDefPoint(5.5,1.5){b}
\tkzDefPoint(4.5,1.5){e}

\tkzDrawSegments(p_1,p_2 p_2,p_3 p_3,p_4 p_4,p_5 p_5,p_6)
\tkzDrawSegments[thick,gray](p_1,p_6)
\tkzDrawSegments[thick,dotted](p_2,s_2 p_3,s_3 p_4,s_4 p_5,s_5 a,b)
\tkzDrawPoints(p_1,p_2,p_3,p_4,p_5,p_6,s_2,s_3,s_4,s_5)
\tkzLabelPoint[above](e){\(\varepsilon\)}
\tkzLabelPoints(p_1,p_3,p_5,p_6,s_3)
\tkzLabelPoints[right](p_2,p_4)
\tkzLabelPoints[above left](s_4)
\tkzLabelPoints[below left](s_5)
\tkzLabelPoints[below](s_2)
\end{tikzpicture}
\end{minipage}
\caption{Left: Alignment for the Hausdorff distance.
Right: Alignment for the Fr\'echet distance.
In both cases, the alignment is described as the sequence
\(\langle s_1 \coloneqq p_1, s_2, s_3, s_4, s_5, s_6 \coloneqq p_6\rangle\).}
\label{fig:alignments}
\end{figure}
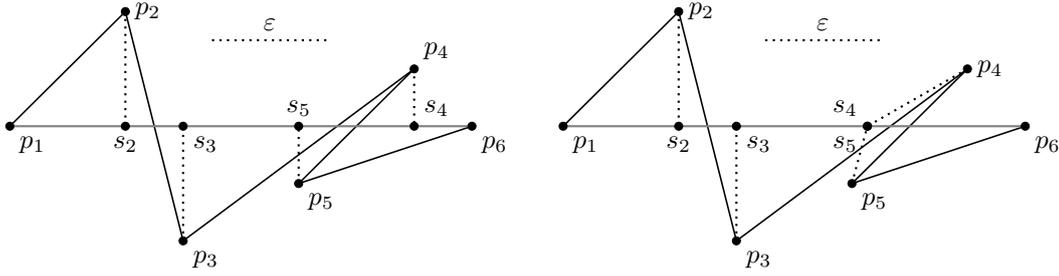

The \emph{Fr\'echet distance} is often described through an analogy with a person and a dog walking
along their respective curves without backtracking, where the Fr\'echet distance is the shortest
leash needed for such a walk.
Formally, consider a set of \emph{reparametrisations} \(\Phi_\ell\) of length \(\ell\), defined as
continuous non-decreasing surjective functions \(\phi: [0, 1] \to [1, \ell]\).
Given two polygonal curves \(\pi\) and \(\sigma\) of lengths \(m\) and \(n\), respectively,
we can define the Fr\'echet distance as
\[\fr(\pi, \sigma) \eqdef \adjustlimits\inf_{\alpha \in \Phi_m, \beta \in \Phi_n} \max_{t \in [0, 1]}
\lVert\pi(\alpha(t)) - \sigma(\beta(t))\rVert\,.\]
We refer to the pair of reparametrisations as an \emph{alignment.}
We often consider the Fr\'echet distance between a curve \(\pi = \langle p_1, \dots, p_n\rangle\)
and a line segment \(p_1p_n\), for some \(n \in \N\), \(n \geq 3\).
In this setting the alignment can be described in a more intuitive way; see also
\cref{fig:alignments}.
It can be described as a sequence of locations on the line segment to which
the vertices of the curves are matched, \(\langle s_2, \dots, s_{n - 1}\rangle\), where
\(s_i \in [1, 2]\) for all \(i \in \{2, \dots, n - 1\}\) and \(s_i \leq s_{i + 1}\) for all
\(i \in \{2, \dots, n - 2\}\).
To see that, assign \(s_1 \coloneqq 1\) and \(s_n \coloneqq 2\) and construct a helper
reparametrisation \(\phi: [0, 1] \to [1, n]\), defined as
\(\phi(t) = (n - 1) \cdot t + 1\) for any \(t \in [0, 1]\).
Construct another reparametrisation \(\psi: [1, n] \to [1, 2]\), defined as
\[\psi(t) = \begin{cases}
s_{\lfloor t\rfloor} \cdot (1 - t + \lfloor t\rfloor)
+ s_{\lfloor t\rfloor + 1} \cdot (t - \lfloor t\rfloor) & \text{if } t \in [1, n),\\
s_n & \text{if } t = n.
\end{cases}\]
Note that \(\phi\) and \(\psi \circ \phi\) satisfy the definition of reparametrisations for \(\pi\)
and \(p_1p_n\), respectively.

We also define an \emph{alignment} between a curve and a line segment for the Hausdorff distance
(see \cref{fig:alignments}).
It represents the map from the curve to the line segment, where each point on the curve is mapped to
the closest point on the line segment.
It is given by a sequence \(\langle s_1, \dots, s_n\rangle\), where \(s_i \in [1, 2]\) for all
\(i \in [n]\), such that \(p_1p_n(s_i) = \argmin_{p' \in p_1p_n} \lVert p' - p_i\rVert\).
In other words, \(p_1p_n(s_i)\) is the closest point to \(p_i\) for all \(i \in [n]\); as we show in
\cref{sec:intermediate_hausdorff}, the Hausdorff distance is realised as the distance between
\(p_i\) and \(p_1p_n(s_i)\) for some \(i \in [n]\).
Therefore, establishing such an alignment and checking that
\(\lVert p_1p_n(s_i) - p_i\rVert \leq \varepsilon\) for all \(i \in [n]\) allows us to check that
\(\hs(\pi, p_1p_n) \leq \varepsilon\) for some \(\varepsilon \in \R^{\geq 0}\).

We are discussing the following problem: given an uncertain curve
\(\U = \langle U_1, \dots, U_n\rangle\) with \(n \in \N\), \(n \geq 3\), and \(U_i \subset \Rt\) for
all \(i \in [n]\), and the threshold \(\varepsilon \in \R^{> 0}\), find a minimal-length
subsequence \(\U' = \langle U_{i_1}, \dots, U_{i_\ell}\rangle\) of \(\U\) with \(\ell \leq n\), such
that for any realisation \(\pi \Subset \U\), the corresponding realisation \(\pi' \Subset \U'\)
forms an \(\varepsilon\)-simplification of \(\pi\) under some distance measure \(\delta\).
We solve this problem both for the Hausdorff and the Fr\'echet distance for uncertainty
modelled with indecisive points, line segments, disks,
and convex polygons.

\section{Overview of the Approach}\label{sec:overview}
In this \lcnamecref{sec:overview}, we present the short description of our approach in different
settings.
We work out the details and show correctness in \cref{sec:intermediate,sec:shortcut,sec:graph}.

On the highest level, we use the \emph{shortcut graph.}
Each uncertain point of an uncertain curve \(\U = \langle U_1, \dots, U_n\rangle\) corresponds to a
vertex.
An edge connects two vertices \(i\) and \(j\) if and only if the distance between any realisation of
\(\U[i: j]\) and the corresponding line segment from \(U_i\) to \(U_j\) is below the threshold.
The path with the least edges from vertex \(1\) to vertex \(n\) then corresponds to the
simplification using least uncertain points.
So, we construct the shortcut graph and find the shortest path between two vertices.
The key idea is that we find shortcuts that are valid for \emph{all} realisations, so any sequence
of shortcuts can be chosen.
We discuss this in \cref{sec:graph}.

In order to construct the shortcut graph, we need to check whether an edge should be added to the
graph, i.e.\@ whether a shortcut is \emph{valid.}
The approach is different for the Hausdorff and the Fr\'echet distance and for each uncertainty
model.
For the first and the last uncertain point of the shortcut, we state in \cref{sec:shortcut} that
there are several critical pairs of realisations that need to be tested explicitly, and then for any
other pair of realisations, we know that the distance is also below the threshold.
Testing each pair corresponds to finding the distance between a precise line segment and any
realisation of an uncertain curve; we show the simple procedures to do this in
\cref{sec:intermediate}.

\section{Shortcut Testing: Intermediate Points}\label{sec:intermediate}
In this \lcnamecref{sec:intermediate}, we discuss testing a single shortcut where we fix the
realisations of the first and the last uncertain point.
We start by showing some basic facts about the Hausdorff and the Fr\'echet distance in the precise
setting, and then we use them to design simple algorithms for testing shortcuts in the uncertain
settings.
We answer the following problem.
\begin{problem}\label{prob:intermediate}
Given an uncertain curve \(\U = \langle U_1, \dots, U_n\rangle\) on \(n \in \N\), \(n \geq 3\)
uncertain points in \(\Rt\), as well as realisations \(p_1 \in U_1\), \(p_n \in U_n\), check if the
largest Hausdorff or Fr\'echet distance between \(\U\) and its one-segment simplification is below a
threshold \(\varepsilon \in \R^{> 0}\) for any realisation with the fixed start and end points,
i.e.\@ for \(\delta \coloneqq \hs\) or \(\delta \coloneqq \fr\), verify
\[\max_{\pi \Subset \U, \pi(1) \equiv p_1, \pi(n) \equiv p_n} \delta(\pi, p_1p_n)
\leq \varepsilon\,.\]
\end{problem}

\subsection{Hausdorff Distance}\label{sec:intermediate_hausdorff}
We start by showing some useful facts about the Hausdorff distance in the precise setting.
We then solve \cref{prob:intermediate} for \(\delta \coloneqq \hs\).
\begin{lemma}\label{lem:hausdorff_sub}
Given \(n \in \N^{> 0}\) and a precise curve \(\pi = \langle p_1, \dots, p_n\rangle\)
with \(p_i \in \Rt\) for all \(i \in [n]\), we have that for any \(q \in p_1p_n\) with
\(q \prec p_n\), there is some \(i \in [n - 1]\) such that
\(s \coloneqq \argmin_{r \in p_1p_n} \lVert p_i - r\rVert\),
\(t \coloneqq \argmin_{r \in p_1p_n} \lVert p_{i + 1} - r\rVert\), and \(s \preccurlyeq q \prec t\).
\end{lemma}
\begin{proof}
Assume this is not the case and pick a point \(q \in p_1p_n \setminus p_n\) that forms a
counterexample.
We now have for all \(i \in [n - 1]\) and the definitions of \(s\) and \(t\) given above, that
\(s \preccurlyeq q \implies t \preccurlyeq q\).
Clearly, for \(i = 1\) we have \(s \equiv p_1\) and so \(s \preccurlyeq q\).
By induction on \(i\), we can conclude that for all \(i \in [n]\),
\(\argmin_{r \in p_1p_n} \lVert p_i - r\rVert \preccurlyeq q\).
In particular, as \(\argmin_{r \in p_1p_n} \lVert p_n - r\rVert = p_n\), this means that
\(p_n \preccurlyeq q\).
However, we picked \(q \prec p_n\).
This is a contradiction, so the lemma holds.
\end{proof}

\begin{lemma}\label{lem:distance_segm}
Given four points \(a, b, c, d \in \Rt\) forming segments \(ab\) and \(cd\), the highest distance
from one segment to the other is achieved at an endpoint:
\[\max_{p \in ab} d(p, cd) = \max\big\{d(a, cd), d(b, cd)\big\}\,.\]
\end{lemma}
\begin{proof}
As \(a, b \in ab\), trivially we get
\(\max_{p \in ab} d(p, cd) \geq \max\big\{d(a, cd), d(b, cd)\big\}\),
so it remains to show that \(\max_{p \in ab} d(p, cd) \leq \max\big\{d(a, cd), d(b, cd)\big\}\).
Consider two sets \(S_1 \coloneqq \{p \mid \lVert p\rVert \leq \varepsilon\}\) and
\(S_2 \coloneqq cd\), with \(\varepsilon \coloneqq \max\big\{d(a, cd), d(b, cd)\big\}\).
Take their Minkowski sum:
\begin{align*}
S &\coloneqq \{p + q \mid p \in S_1, q \in S_2\}\\
&\eqdef \{p + q \mid \lVert p\rVert \leq \varepsilon, q \in cd\}\\
&= \{r \mid \lVert r - q\rVert \leq \varepsilon, q \in cd\}\\
&= \{r \mid \min_{q \in cd} \lVert r - q\rVert \leq \varepsilon\}\\
&\eqdef \{r \mid d(r, cd) \leq \varepsilon\}\,.
\end{align*}
Note that both sets are convex: \(S_1\) is a disk and \(S_2\) is a line segment.
Then their Minkowski sum \(S\) is also convex.
By definition of \(S\) and \(\varepsilon\), we have \(a, b \in S\).
By definition of a convex set, we conclude that \(ab \in S\), so
\(\max_{p \in ab} d(p, cd) \leq \max\big\{d(a, cd), d(b, cd)\big\}\), and the statement of the
lemma holds.
\end{proof}

\begin{lemma}\label{lem:hausdorff_precise}
Given \(n \in \N^{> 0}\), for any precise curve \(\pi = \langle p_1, \dots, p_n\rangle\)
with \(p_i \in \Rt\) for all \(i \in [n]\), we have
\[\hs(\pi, p_1p_n) = \max_{i \in [n]} d(p_i, p_1p_n)\,.\]
\end{lemma}
\begin{proof}
Recall the definition of Hausdorff distance in this setting:
\[\hs(\pi, p_1p_n) = \max\Big\{\max_{p \in \pi} \min_{q \in p_1p_n} \lVert p - q\rVert,
\max_{q \in p_1p_n} \min_{p \in \pi} \lVert p - q\rVert\Big\}\,.\]

We first show that \(\max_{q \in p_1p_n} \min_{p \in \pi} \lVert p - q\rVert \leq
\max_{p \in \pi} \min_{q \in p_1p_n} \lVert p - q\rVert \eqqcolon \varepsilon\).
We do a case distinction on \(q \in p_1p_n\) and show that for all \(q\), we have
\(\min_{p \in \pi} \lVert p - q\rVert \leq \varepsilon\).
\begin{itemize}
    \item \(q \equiv p_n\).
    Note \(p_n \in \pi\), so \(\min_{p \in \pi} \lVert p - q\rVert = 0 \leq \varepsilon\).
    \item \(q \prec p_n\).
    Using \cref{lem:hausdorff_sub}, we can find \(i \in [n - 1]\) and the corresponding \(s\) and
    \(t\) such that \(s \preccurlyeq q \prec t\).
    As \(\max_{p \in \pi} d(p, p_1p_n) \eqdef \varepsilon\),
    \(d(p_i, p_1p_n) = \lVert p_i - s\rVert \leq \varepsilon\) and
    \(d(p_{i + 1}, p_1p_n) = \lVert p_{i + 1} - t\rVert \leq \varepsilon\).
    But then also \(d(s, p_ip_{i + 1}) \leq \lVert s - p_i\rVert \leq \varepsilon\) and
    \(d(t, p_ip_{i + 1}) \leq \lVert t - p_{i + 1}\rVert \leq \varepsilon\).
    By \cref{lem:distance_segm}, we conclude that
    \(\max_{r \in st} d(r, p_ip_{i + 1}) \leq \varepsilon\).
    As \(s \preccurlyeq q \prec t\), we have \(q \in st\), so
    \(d(q, p_ip_{i + 1}) \leq \varepsilon\).
\end{itemize}
This covers all cases, so indeed for all \(q \in p_1p_n\),
\(\min_{p \in \pi} \lVert p - q\rVert \leq \varepsilon\), and hence we conclude
\(\max_{q \in p_1p_n} \min_{p \in \pi} \lVert p - q\rVert \leq
\max_{p \in \pi} \min_{q \in p_1p_n} \lVert p - q\rVert\).
We can derive
\begin{align*}
\hs(\pi, p_1p_n) &= \max\Big\{\max_{p \in \pi} \min_{q \in p_1p_n} \lVert p - q\rVert,
\max_{q \in p_1p_n} \min_{p \in \pi} \lVert p - q\rVert\Big\}\\
&= \max_{p \in \pi} \min_{q \in p_1p_n} \lVert p - q\rVert\\
&\eqdef \max_{p \in \pi} d(p, p_1p_n)\\
&\eqdef \max_{i \in [n - 1]} \max_{p \in p_ip_{i + 1}} d(p, p_1p_n)\\
&\reason{\cref{lem:distance_segm}}\\
&= \max_{i \in [n - 1]} \max\big\{d(p_i, p_1p_n), d(p_{i + 1}, p_1p_n)\big\}\\
&\eqdef \max_{i \in [n]} d(p_i, p_1p_n)\,,
\end{align*}
as was to be shown.
\end{proof}

\paragraph*{Indecisive points.}
We are now ready to generalise the setting to include imprecision.
We first show that the straightforward setting with indecisive points permits an easy solution using
\cref{lem:hausdorff_precise}.
\begin{lemma}\label{lem:hausdorff_ind}
Given \(n, k \in \N^{> 0}\), \(n \geq 3\), for any indecisive curve
\(\U = \langle U_1, \dots, U_n\rangle\) with \(U_i = \{p_i^1, \dots, p_i^k\}\) for all \(i \in [n]\)
and \(p_i^j \in \Rt\) for all \(i \in [n]\), \(j \in [k]\), and given some \(p_1 \in U_1\) and
\(p_n \in U_n\), we have
\[\max_{\pi \Subset \U, \pi(1) \equiv p_1, \pi(n) \equiv p_n} \hs(\pi, p_1p_n) =
\max_{i \in \{2, \dots, n - 1\}} \max_{j \in [k]} d(p_i^j, p_1p_n)\,.\]
\end{lemma}
\begin{proof}
Assume the setting of the lemma statement.
Derive
\begin{align*}
&\max_{\pi \Subset \U, \pi(1) \equiv p_1, \pi(n) \equiv p_n} \hs(\pi, p_1p_n)\\
&\reason{\cref{lem:hausdorff_precise}}\\
&= \max_{\pi \Subset \U, \pi(1) \equiv p_1, \pi(n) \equiv p_n} \max_{i \in [n]} d(\pi(i), p_1p_n)\\
&\reason{Def.~\(\Subset\), \(d(p_1, p_1p_n) = d(p_n, p_1p_n) = 0\)}\\
&= \max_{i \in \{2, \dots, n - 1\}} \max_{p \in U_i} d(p, p_1p_n)\\
&\reason{Def.~indecisive point}\\
&= \max_{i \in \{2, \dots, n - 1\}} \max_{j \in [k]} d(p_i^j, p_1p_n)\,,
\end{align*}
as was to be shown.
\end{proof}
Note that this means that when the start and end realisations are fixed, we can test that a
shortcut is valid using the lemma above in time \(\bO(nk)\) for a shortcut of length \(n\).

\paragraph*{Disks.}
We proceed to present the way to test shortcuts for fixed realisations of the first and the last
points when the imprecision is modelled using disks.
In the next arguments the following well-known form of a triangle inequality is useful.
\begin{lemma}\label{lem:triangle_sets}
Given a metric space \((X, d)\) and a non-empty subset \(S \subset X\), \(S \neq \emptyset\),
for any \(x, y \in X\),
\[d(x, S) \leq d(x, y) + d(y, S)\,.\]
\end{lemma}
\begin{proof}
Pick some \(z' \in S\) and \(x, y \in X\).
By the triangle inequality, \(d(x, z') \leq d(x, y) + d(y, z')\), so
\[d(x, S) \eqdef \inf_{z \in S} d(x, z) \leq d(x, z') \leq d(x, y) + d(y, z')\,,\]
and this holds for \emph{any} choice of \(z'\).
Therefore, we conclude
\[d(x, S) \leq d(x, y) + \inf_{z \in S} d(y, z) \eqdef d(x, y) + d(y, S)\,.\qedhere\]
\end{proof}
\begin{corollary}\label{cor:triangle_segm}
For any \(p, q \in \Rt\) and a line segment \(ab\) on \(a, b \in \Rt\),
\[d(p, ab) \leq \lVert p - q\rVert + d(q, ab)\,.\]
\end{corollary}

We now state the result for disks.
\begin{lemma}\label{lem:hausdorff_disks}
Given \(n \in \N^{> 0}\), \(n \geq 3\), for any imprecise curve modelled with disks
\(\U = \langle U_1, \dots, U_n\rangle\) with \(U_i = D(c_i, r_i)\) for all \(i \in [n]\)
and \(c_i \in \Rt\), \(r_i \in \R^{\geq 0}\) for all \(i \in [n]\), and given some \(p_1 \in U_1\)
and \(p_n \in U_n\), we have
\[\max_{\pi \Subset \U, \pi(1) \equiv p_1, \pi(n) \equiv p_n} \hs(\pi, p_1p_n) =
\max_{i \in \{2, \dots, n - 1\}} \big(d(c_i, p_1p_n) + r_i\big)\,.\]
\end{lemma}
\begin{proof}
Assume the setting of the lemma.
As before, we derive
\begin{align*}
&\max_{\pi \Subset \U, \pi(1) \equiv p_1, \pi(n) \equiv p_n} \hs(\pi, p_1p_n)\\
&\reason{\cref{lem:hausdorff_precise}}\\
&= \max_{\pi \Subset \U, \pi(1) \equiv p_1, \pi(n) \equiv p_n} \max_{i \in [n]} d(\pi(i), p_1p_n)\\
&\reason{Def.~\(\Subset\), \(d(p_1, p_1p_n) = d(p_n, p_1p_n) = 0\)}\\
&= \max_{i \in \{2, \dots, n - 1\}} \max_{p \in U_i} d(p, p_1p_n)\,.
\end{align*}
It remains to show that \(\max_{p \in U_i} d(p, p_1p_n) = d(c_i, p_1p_n) + r_i\) for any
\(i \in \{2, \dots, n - 1\}\).

Firstly, pick \(p' \coloneqq \argmax_{p \in U_i} d(p, p_1p_n)\).
Note that by \cref{cor:triangle_segm},
\(d(p', p_1p_n) \leq \lVert p' - c_i\rVert + d(c_i, p_1p_n)\).
Furthermore, as \(p' \in U_i\), by definition of \(U_i\) we have \(\lVert p' - c_i\rVert \leq r_i\).
Thus, \(\max_{p \in U_i} d(p, p_1p_n) \leq d(c_i, p_1p_n) + r_i\), and it remains to show the
inequality in the other direction.

Now pick a point \(q' \coloneqq \argmin_{q \in p_1p_n} \lVert q - c_i\rVert\), so that
\(d(c_i, p_1p_n) = \lVert q' - c_i\rVert\).
Draw the line through \(c_i\) and \(q'\) and pick the point \(p'\) on that line on the boundary of
\(U_i\) on the opposite side of \(q\) w.r.t.\@ \(c_i\).
Clearly, \(\lVert p' - c_i\rVert = r_i\) and \(q' = \argmin_{q \in p_1p_n} \lVert q - p'\rVert\).
Thus,
\[d(p', p_1p_n) = \lVert p' - q'\rVert = \lVert q' - c_i\rVert + \lVert p' - c_i\rVert
= d(c_i, p_1p_n) + r_i\,.\]
Note that \(p' \in U_i\), so we conclude
\(\max_{p \in U_i} d(p, p_1p_n) \geq d(c_i, p_1p_n) + r_i\).
Hence, the statement of the lemma holds.
\end{proof}
Once again, note that this lemma allows us to test a shortcut in a straightforward manner, in time
\(\bO(n)\) for a shortcut of length \(n\).

\paragraph*{Polygonal closed convex sets.}
\begin{lemma}\label{lem:hausdorff_pccs}
Given \(n, k \in \N^{> 0}\), \(n \geq 3\), for any imprecise curve modelled with PCCSs
\(\U = \langle U_1, \dots, U_n\rangle\) with \(U_i \subset \Rt\)
and \(V(U_i) = \{p_i^1, \dots, p_i^k\}\) for all \(i \in [n]\), and given some
\(p_1 \in U_1\) and \(p_n \in U_n\), we have
\[\max_{\pi \Subset \U, \pi(1) \equiv p_1, \pi(n) \equiv p_n} \hs(\pi, p_1p_n) =
\max_{i \in \{2, \dots, n - 1\}} \max_{v \in V(U_i)} d(v, p_1p_n)\,.\]
\end{lemma}
\begin{proof}
Assume the setting of the lemma.
As before, derive
\begin{align*}
&\max_{\pi \Subset \U, \pi(1) \equiv p_1, \pi(n) \equiv p_n} \hs(\pi, p_1p_n)\\
&\reason{\cref{lem:hausdorff_precise}}\\
&= \max_{\pi \Subset \U, \pi(1) \equiv p_1, \pi(n) \equiv p_n} \max_{i \in [n]} d(\pi(i), p_1p_n)\\
&\reason{Def.~\(\Subset\), \(d(p_1, p_1p_n) = d(p_n, p_1p_n) = 0\)}\\
&= \max_{i \in \{2, \dots, n - 1\}} \max_{p \in U_i} d(p, p_1p_n)\,.
\end{align*}
To show that the claim holds, it remains to show that for any PCCS \(U\) and a line segment \(ab\)
it holds that \(\max_{p \in U} d(p, ab) = \max_{v \in V(U)} d(v, ab)\).
Firstly, as \(V(U) \subset U\), we immediately have
\(\max_{p \in U} d(p, ab) \geq \max_{v \in V(U)} d(v, ab)\).
Consider any \(p \in U\).
We will show that there is some \(v \in V(U)\) such that \(d(v, ab) \geq d(p, ab)\), thus completing
the proof.
We do a case distinction on \(p\).
\begin{itemize}
    \item \(p \in V(U)\). Then pick \(v \coloneqq p\), and we are done.
    \item \(p \notin V(U)\), but \(p\) is on the boundary of \(U\).
    Consider the vertices \(v, w \in V(U)\) with \(p \in vw\).
    Using \cref{lem:distance_segm}, we note
    \[\max_{q \in vw} d(q, ab) = \max\big\{d(v, ab), d(w, ab)\big\}\,.\]
    W.l.o.g.\@ suppose \(d(v, ab) \geq d(w, ab)\).
    Then for \(v\) indeed we have \(d(v, ab) \geq d(p, ab)\).
    \item \(p\) is in the interior of \(U\) (cannot occur for line segments).
    Find the point \(q' \coloneqq \argmin_{q \in ab} \lVert p - q\rVert\), so
    \(d(p, ab) = \lVert p - q'\rVert\).
    Draw the line through \(p\) and \(q'\); let \(p'\) be the point on that line on the boundary of
    \(U\) on the opposite side of \(q'\) w.r.t.\@ \(p\).
    Clearly, \(q' = \argmin_{q \in ab} \lVert p' - q\rVert\), so \(d(p', ab) > d(p, ab)\).
    Then we can find a vertex \(v \in V(U)\) as in the previous cases, yielding
    \(d(v, ab) \geq d(p', ab) > d(p, ab)\).
\end{itemize}
This covers all the cases, so the statement holds.
\end{proof}
As before, this lemma gives us a simple way to test the shortcut with fixed realisations of the
first and the last points in time \(\bO(nk)\) for a shortcut of length \(n\) and PCCSs with \(k\)
vertices.

\subsection{Fr\'echet Distance}\label{sec:intermediate_frechet}
We now turn our attention to the Fr\'echet distance.
In this \lcnamecref{sec:intermediate_frechet}, we do not show results for the Fr\'echet distance in
the precise setting.
For extra intuition, we show \cref{alg:precise_frechet}, which follows from a well-known fact shown
e.g.\@ by Guibas et al.~\cite[Lemma~8]{guibas}; it can also be seen as specialisation of the
indecisive point case to \(k = 1\) or of the disk case to \(r = 0\).

\begin{algorithm}[tbp]
\caption{Testing a shortcut on a precise curve with the Fr\'echet distance.}
\label{alg:precise_frechet}
\begin{algorithmic}[1]
\Require{$\pi = \langle p_1, \dots, p_n\rangle$, $n \in \N^{> 0}$, $\forall i \in [n]: p_i \in \Rt$,
$\varepsilon \in \R^{> 0}$}
\Function{CheckFr\'echetPrecise}{$\pi, n, \varepsilon$}
    \State $s_1 \coloneqq 1$
    \For{$i \in \{2, \dots, n - 1\}$}
        \State $S_i \coloneqq \{t \in [s_{i - 1}, 2] \mid \lVert p_i - p_1p_n(t)\rVert \leq \varepsilon\}$
        \If{$S_i = \emptyset$}
            \State\Return\False
        \EndIf
        \State $s_i \coloneqq \min S_i$
    \EndFor
    \State\Return\True
\EndFunction
\end{algorithmic}
\end{algorithm}

\paragraph*{Indecisive points.}
The idea is that in the precise case we can always align greedily as we move along the line segment.
In this case, we also need to find the realisation for each indecisive point that makes for the
`worst' greedy choice.

\begin{algorithm}[tbp]
\caption{Testing a shortcut on an indecisive curve with the Fr\'echet distance.}
\label{alg:indecisive_frechet}
\begin{algorithmic}[1]
\Require{$\U = \langle U_1, \dots, U_n\rangle$, $n, k \in \N^{> 0}$,
$\forall i \in [n]: U_i = \{p_i^1, \dots, p_i^k\}$, $\forall i \in [n], j \in [k]: p_i^j \in \Rt$,
$\varepsilon \in \R^{> 0}$, $p_1 \in U_1$, $p_n \in U_n$}
\Function{CheckFr\'echetInd}{$\U, p_1, p_n, n, k, \varepsilon$}
    \State $s_1 \coloneqq 1$
    \For{$i \in \{2, \dots, n - 1\}$}
        \State $T_i \coloneqq \emptyset$
        \For{$j \in [k]$}
            \State $S_i^j \coloneqq \{t \in [s_{i - 1}, 2] \mid \lVert p_i^j - p_1p_n(t)\rVert \leq \varepsilon\}$
            \If{$S_i^j = \emptyset$}
                \State\Return\False
            \EndIf
            \State $T_i \coloneqq T_i \cup \min S_i^j$
        \EndFor
        \State $s_i \coloneqq \max T_i$
    \EndFor
    \State\Return\True
\EndFunction
\end{algorithmic}
\end{algorithm}

\begin{lemma}\label{lem:frechet_indecisive}
Given \(n, k \in \N^{> 0}\) and \(\varepsilon \in \R^{> 0}\), for any indecisive curve
\(\U = \langle U_1, \dots, U_n\rangle\) with \(U_i = \{p_i^1, \dots, p_i^k\}\) for all \(i \in [n]\)
and \(p_i^j \in \Rt\) for all \(i \in [n]\), \(j \in [k]\), and given some \(p_1 \in U_1\) and
\(p_n \in U_n\), we have, using \cref{alg:indecisive_frechet},
\[\max_{\pi \Subset \U, \pi(1) \equiv p_1, \pi(n) \equiv p_n} \fr(\pi, p_1p_n) \leq \varepsilon \iff
\textsl{\textsc{CheckFr\'echetInd}}(\U, p_1, p_n, n, k, \varepsilon) = \True\,.\]
\end{lemma}
\begin{proof}
First, assume that
\(\max_{\pi \Subset \U, \pi(1) \equiv p_1, \pi(n) \equiv p_n} \fr(\pi, p_1p_n) \leq \varepsilon\).
In the \lcnamecref{alg:indecisive_frechet}, we compute some set \(S_i^j\) for each \(p_i^j\) and
then pick one value from it and add it to \(T_i\); from \(T_i\) we then pick a single value as
\(s_i\).
So, \(s_i \in S_i^j\) for some \(j_i \in [k]\), on every iteration \(i \in \{2, \dots, n - 1\}\).
Consider a realisation \(\pi \Subset \U\) with \(\pi(1) \equiv p_1\), \(\pi(n) \equiv p_n\), and
\(\pi(i) \equiv p_i^{j_i}\) for every \(i \in \{2, \dots, n - 1\}\), where \(j_i\) is chosen as the
value corresponding to \(s_i\).
Then we know \(\fr(\pi, p_1p_n) \leq \varepsilon\).
So, there is an alignment that can be given as a sequence of \(n\) positions, \(t_i \in [1, 2]\),
such that \(\lVert \pi(i) - p_1p_n(t_i)\rVert \leq \varepsilon\) and \(t_i \leq t_{i + 1}\) for
all \(i\).
The alignment is established by interpolating between the consecutive points on the curves, as
discussed in \cref{sec:prelims}.

We now show by induction that \(s_i \leq t_i\) for all \(i\).
For \(i = 2\), we get, for the chosen \(j_2\),
\(s_2 \coloneqq \min \{t \in [1, 2] \mid \lVert p_2^{j_2} - p_1p_n(t)\rVert \leq \varepsilon\}\).
As we have \(t_2 \in \{t \in [1, 2] \mid \lVert p_2^{j_i} - p_1p_n(t)\rVert \leq \varepsilon\}\),
we get \(s_2 \leq t_2\).
Now assume the statement holds for some \(i\), then for \(i + 1\) we get
\(s_{i + 1} \coloneqq \min \{t \in [s_i, 2] \mid \lVert p_{i + 1}^{j_{i + 1}} - p_1p_n(t)\rVert
\leq \varepsilon\}\); we can rephrase this so that
\[s_{i + 1} \eqdef \min \big(\{t \in [1, 2] \mid \lVert p_{i + 1}^{j_{i + 1}} - p_1p_n(t)\rVert
\leq \varepsilon\} \cap [s_i, 2]\big)\,.\]
So, there are two options.
\begin{itemize}
\item \(s_{i + 1} = s_i\).
Then we know \(s_{i + 1} = s_i \leq t_i \leq t_{i + 1}\).
\item \(s_{i + 1} > s_i\).
Then we can use the same argument as for \(i = 2\) to find that
\(s_{i + 1} \leq t_{i + 1}\).
\end{itemize}

Now we know that for every \(i\), \(t_i \in S_i^{j_i}\) for the choice of \(j_i\) described above.
Therefore, for any \(p_{i + 1}^{j_{i + 1}}\) there is always a realisation prefix such that any
valid alignment has \(t_{i + 1} \geq s_i\); as we know that there is a valid alignment for every
realisation, we conclude that every \(S_i^j\) is non-empty.
Thus, the \lcnamecref{alg:indecisive_frechet} returns \True.

Now assume that the \lcnamecref{alg:indecisive_frechet} returns \True.
Consider any realisation \(\pi \Subset \U\).
We claim that there is a valid alignment, described with a sequence of \(t_i \in [1, 2]\) for
\(i \in \{2, \dots, n - 1\}\), such that \(s_{i - 1} \leq t_i \leq s_i\) and
\(\lVert p_1p_n(t_i) - \pi(i)\rVert \leq \varepsilon\).
Denote the realisation
\(\pi \eqdef \langle p_1, p_2^{j_2}, p_3^{j_3}, \dots, p_{n - 1}^{j_{n - 1}}, p_n\rangle\),
so the sequence \(\langle j_2, \dots, j_{n - 1}\rangle\) describes the choices of the realisation.
Consider the set \(S_i^{j_i}\) for any \(i \in \{2, \dots, n - 1\}\).
We know that it is non-empty, otherwise the \lcnamecref{alg:indecisive_frechet} would have returned
\False.
We claim that we can pick \(t_i = \min S_i^{j_i}\) for every \(i\).
By definition, \(S_i^{j_i} \subseteq [1, 2]\) and
\(\lVert p_1p_n(t_i) - \pi(i)\rVert \leq \varepsilon\).
We also trivially get that \(s_{i - 1} \leq t_i\).
Finally, note that \(t_i \in T_i\), and \(s_i \coloneqq \max T_i\), so \(t_i \leq s_i\).

This argument shows that \(t_i \leq t_{i + 1}\) for every \(i\), and that
\(\lVert p_1p_n(t_i) - \pi(i)\rVert \leq \varepsilon\).
Therefore, \(\fr(\pi, p_1p_n) \leq \varepsilon\).
As this works for any realisation with \(\pi(1) \equiv p_1\) and \(\pi(n) \equiv p_n\), we conclude
\(\max_{\pi \Subset \U, \pi(1) \equiv p_1, \pi(n) \equiv p_n} \fr(\pi, p_1p_n) \leq \varepsilon\).
\end{proof}

\paragraph*{Disks.}
To show the generalisation to disks, it is helpful to reframe the problem as that of disk stabbing
for appropriate disks.
We demonstrate some useful facts first.

\begin{lemma}\label{lem:aux_disk_point}
Given a disk \(D_1 \coloneqq D(c, r)\) with \(c \in \Rt\), \(r \in \R^{\geq 0}\), a threshold
\(\varepsilon \in \R^{> 0}\), and a point \(p \in \Rt\), define
\(D_2 \coloneqq D(c, \varepsilon - r)\).
We have
\[\max_{p' \in D_1} \lVert p - p'\rVert \leq \varepsilon \iff p \in D_2\,.\]
\end{lemma}
\begin{proof}
First, assume \(p \in D_2 \eqdef \{s \in \Rt \mid \lVert s - c\rVert \leq \varepsilon - r\}\);
thus, we know \(\lVert p - c\rVert \leq \varepsilon - r\).
Take \(q \coloneqq \argmax_{p' \in D_1} \lVert p - p'\rVert\).
Then \(q \in D_1 \eqdef \{s \in \Rt \mid \lVert s - c\rVert \leq r\}\), so
\(\lVert q - c\rVert \leq r\).
Then by the triangle inequality,
\[\lVert p - q\rVert \leq \lVert p - c\rVert + \lVert q - c\rVert
\leq \varepsilon - r + r = \varepsilon\,.\]

Now assume that \(p \notin D_2 \eqdef \{s \in \Rt \mid \lVert s - c\rVert \leq \varepsilon - r\}\).
Then \(\lVert p - c\rVert > \varepsilon - r\).
Consider a point \(q\) on the line \(pc\) on the boundary of \(D_1\), so that \(c\) is between \(p\)
and \(q\) on the line.
Note that \(q \in D_1\), so
\[\max_{p' \in D_1} \lVert p - p'\rVert \geq \lVert p - q\rVert
= \lVert p - c\rVert + \lVert q - c\rVert > \varepsilon - r + r = \varepsilon\,,\]
completing the proof.
\end{proof}

We can now generalise the previous statement to talk about distance to line segments.
\begin{lemma}\label{lem:aux_disk_segm}
Given a disk \(D_1 \coloneqq D(c, r)\) with \(c \in \Rt\), \(r \in \R^{\geq 0}\), a threshold
\(\varepsilon \in \R^{> 0}\), and a line segment \(pq\) with \(p, q \in \Rt\), define
\(D_2 \coloneqq D(c, \varepsilon - r)\).
We have
\[\max_{p' \in D_1} d(p', pq) \leq \varepsilon \iff pq \cap D_2 \neq \emptyset\,.\]
\end{lemma}
\begin{proof}
First, assume \(pq \cap D_2 \neq \emptyset\).
Take \(t \in pq \cap D_2\).
Consider an arbitrary point \(s \in D_1\).
By \cref{lem:aux_disk_point}, we know that \(\lVert t - s\rVert \leq \varepsilon\); so also
\(d(s, pq) \eqdef \min_{q' \in pq} \lVert q' - s\rVert \leq \lVert t - s\rVert \leq \varepsilon\).
As this holds for arbitrary \(s \in D_1\), we conclude
\(\max_{p' \in D_1} \min_{q' \in pq} \lVert p' - q'\rVert \leq \varepsilon\).

Now assume that \(\max_{p' \in D_1} d(p', pq) \leq \varepsilon\).
Take \(s \coloneqq \argmax_{p' \in D_1} \min_{q' \in pq} \lVert p' - q'\rVert\) and
\(t \coloneqq \argmin_{q' \in pq} \lVert s - q'\rVert\).
In disks it is easy to see that the furthest point of a disk from a line segment is positioned in a
way that the centre of the disk is on the line through the point of the disk and the closest point
of the line segment, so in our case \(c \in st\).
Then \(\lVert t - c\rVert = \lVert t - s\rVert - \lVert s - c\rVert \leq \varepsilon - r\), so
indeed \(t \in D_2\), and \(pq \cap D_2 \neq \emptyset\).
\end{proof}

\begin{algorithm}[tbp]
\caption{Testing a shortcut on an imprecise curve modelled with disks with the Fr\'echet distance.}
\label{alg:disks_frechet}
\begin{algorithmic}[1]
\Require{$\U = \langle U_1, \dots, U_n\rangle$, $n \in \N^{> 0}$,
$\forall i \in [n]: U_i = D(c_i, r_i)$, $\forall i \in [n]: c_i \in \Rt, r_i \in \R^{\geq 0}$,
$\varepsilon \in \R^{> 0}$, $p_1 \in U_1$, $p_n \in U_n$}
\Function{CheckFr\'echetDisks}{$\U, p_1, p_n, n, \varepsilon$}
    \State $s_1 \coloneqq 1$
    \For{$i \in \{2, \dots, n - 1\}$}
        \State $S_i \coloneqq \{t \in [s_{i - 1}, 2] \mid \lVert c_i - p_1p_n(t)\rVert \leq \varepsilon - r_i\}$
        \If{$S_i = \emptyset$}
            \State\Return\False
        \EndIf
        \State $s_i \coloneqq \min S_i$
    \EndFor
    \State\Return\True
\EndFunction
\end{algorithmic}
\end{algorithm}

\begin{lemma}\label{lem:frechet_disks}
Given \(n \in \N^{> 0}\) and \(\varepsilon \in \R^{> 0}\), for any imprecise curve modelled
with disks \(\U = \langle U_1, \dots, U_n\rangle\) with \(U_i = D(c_i, r_i)\) for all \(i \in [n]\)
and \(c_i \in \Rt\), \(r_i \in \R^{\geq 0}\) for all \(i \in [n]\), and given some \(p_1 \in U_1\)
and \(p_n \in U_n\), we have, using \cref{alg:disks_frechet},
\[\max_{\pi \Subset \U, \pi(1) \equiv p_1, \pi(n) \equiv p_n} \fr(\pi, p_1p_n) \leq \varepsilon \iff
\textsl{\textsc{CheckFr\'echetDisks}}(\U, p_1, p_n, n, \varepsilon) = \True\,.\]
\end{lemma}
\begin{proof}
It is convenient to use \cref{lem:aux_disk_segm} to change the problem: rather than establishing an
alignment that comes in the correct order and satisfies the distance constraints, we can do disk
stabbing and pick the stabbing points in the correct order.
So, we have
\(\max_{\pi \Subset \U, \pi(1) \equiv p_1, \pi(n) \equiv p_n} \fr(\pi, p_1p_n) \leq \varepsilon\)
if and only if there exists a sequence of points \(p'_i \in p_1p_n \cap D(c_i, \varepsilon - r_i)\)
for all \(i \in \{2, \dots, n - 1\}\) such that \(p'_i \preccurlyeq p'_{i + 1}\) along \(p_1p_n\)
for all \(i \in \{2, \dots, n - 2\}\).
It remains to show that this is exactly what \cref{alg:disks_frechet} computes.

Assume the \lcnamecref{alg:disks_frechet} returns \True.
We claim that in this case the alignment obtained by \(p'_i \coloneqq p_1p_n(s_i)\) satisfies the
conditions.
First, by definition, \(s_i \in S_i \eqdef
\{t \in [s_{i - 1}, 2] \mid \lVert c_i - p_1p_n(t)\rVert \leq \varepsilon - r_i\}\), so we have
\(\lVert c_i - p'_i\rVert \leq \varepsilon - r_i\), so indeed
\(p'_i \in p_1p_n \cap D(c_i, \varepsilon - r_i)\).
Furthermore, by construction, \(s_i \in [s_{i - 1}, 2]\), so \(s_{i - 1} \leq s_i\), and hence
\(p'_{i - 1} \preccurlyeq p'_i\).

Now assume that the conditions hold, so there is some valid alignment, represented by a sequence
of points \(p'_i\).
We claim that for all \(i \in \{2, \dots, n - 1\}\), we have \(p_1p_n(s_i) \preccurlyeq p'_i\).
For \(i = 2\), this clearly holds, as \(p_1p_n(s_2)\) is the first point that falls into
\(p_1p_n \cap D(c_2, \varepsilon - r_2)\).
Now assume this holds for some \(i\), and we will show that it holds for iteration \(i + 1\).
On iteration \(i + 1\), there are two possibilities:
\begin{itemize}
\item \(s_i > s_{i - 1}\); then we are in the same situation as for \(i = 2\), so
\(p_1p_n(s_i) \preccurlyeq p'_i\).
\item \(s_i = s_{i - 1}\); then we immediately get the same result, as also
\(p'_i \preccurlyeq p'_{i + 1}\).
\end{itemize}
Therefore, we can conclude that the \lcnamecref{alg:disks_frechet} finds an alignment if one exists,
as all \(t_i\) such that \(p_1p_n(t_i) \equiv p'_i\) fall inside \(S_i\), so all \(S_i\) are
non-empty, and the algorithm returns \True.
\end{proof}

\begin{algorithm}[tbp]
\caption{Testing a shortcut on an imprecise curve modelled with PCCSs with the Fr\'echet
distance.}
\label{alg:pccs_frechet}
\begin{algorithmic}[1]
\Require{$\U = \langle U_1, \dots, U_n\rangle$, $n, k \in \N^{> 0}$,
$\forall i \in [n]: U_i \text{ is a PCCS}, V(U_i) = \{p_i^1, \dots, p_i^k\}$,
$\forall i \in [n], j \in [k]: p_i^j \in \Rt$, $\varepsilon \in \R^{> 0}$,
$p_1 \in U_1$, $p_n \in U_n$}
\Function{CheckFr\'echetPCCS}{$\U, p_1, p_n, n, k, \varepsilon$}
    \State $s_1 \coloneqq 1$
    \For{$i \in \{2, \dots, n - 1\}$}
        \State $T_i \coloneqq \emptyset$
        \For{$j \in [k]$}
            \State $S_i^j \coloneqq \{t \in [s_{i - 1}, 2] \mid \lVert p_i^j - p_1p_n(t)\rVert \leq \varepsilon\}$
            \If{$S_i^j = \emptyset$}
                \State\Return\False
            \EndIf
            \State $T_i \coloneqq T_i \cup \min S_i^j$
        \EndFor
        \State $s_i \coloneqq \max T_i$
    \EndFor
    \State\Return\True
\EndFunction
\end{algorithmic}
\end{algorithm}

\paragraph*{Polygonal closed convex sets.}
\begin{lemma}\label{lem:frechet_pccs}
Given \(n, k \in \N^{> 0}\) and \(\varepsilon \in \R^{> 0}\), for any imprecise curve modelled
with PCCSs \(\U = \langle U_1, \dots, U_n\rangle\) with \(U_i \subset \Rt\) and
\(V(U_i) = \{p_i^1, \dots, p_i^k\}\) for all \(i \in [n]\), and given some \(p_1 \in U_1\) and
\(p_n \in U_n\), we have, using \cref{alg:pccs_frechet},
\[\max_{\pi \Subset \U, \pi(1) \equiv p_1, \pi(n) \equiv p_n} \fr(\pi, p_1p_n) \leq \varepsilon \iff
\textsl{\textsc{CheckFr\'echetPCCS}}(\U, p_1, p_n, n, k, \varepsilon) = \True\,.\]
\end{lemma}
\begin{proof}
As we have shown in \cref{lem:hausdorff_pccs}, it suffices to test the vertices of a PCCS to
establish that the distance from every point to the line segment is below the threshold.
It remains to show that the extreme alignment (in terms of ordering) for the Fr\'echet distance is
also achieved at a vertex.
This case then becomes identical to the indecisive points case, so we can use
\cref{lem:frechet_indecisive} to show correctness.

\begin{figure}
\centering
\begin{tikzpicture}[scale=2]
\tkzDefPoints{0/0/p,2/0/q,1/0/x',0/.5/u,2/1/v,.5/0/x,1/.75/t}
\tkzDrawSegments(u,v p,q)
\tkzDrawSegments[dotted](t,x' t,x)
\tkzDrawPoints(x,x',t)
\tkzLabelPoints[below](x',u,v,p,q)
\tkzLabelPoint[below](x){\(x\phantom{'}\)}
\tkzLabelSegment[right](x',t){\(y'\)}
\tkzLabelSegment[left](x,t){\(\varepsilon\)}
\end{tikzpicture}
\caption{Illustration for the computation in \cref{lem:frechet_pccs}.}
\label{fig:computation}
\end{figure}
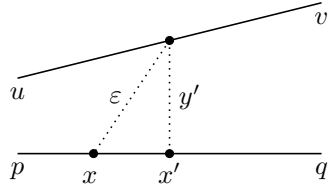

Consider an arbitrary point \(t \in U_i\) and let \(s\) be the earliest point in the
\(\varepsilon\)-disk around \(t\) that is on \(pq\).
Clearly, if \(t\) is in the interior of \(U_i\), then we can take any \(t'\) on the line through
\(t\) parallel to \(pq\) and get the corresponding \(s'\) with \(s \prec s'\).
So, assume \(t\) is on the boundary of \(U_i\).
Suppose that \(t \in uv\) with \(u, v \in V(U_i)\).
Rotate and translate the coordinate plane so that \(pq\) lies on the \(x\)-axis.
Derive the equation for the line containing \(uv\), say, \(y' = kx' + b\).
First consider \(k = 0\), so the line containing \(uv\) is parallel to the line containing \(pq\).
In this case, clearly, moving along \(uv\) in the direction coinciding with the direction from \(p\)
to \(q\) increases the \(x\)-coordinate of point of interest, so moving to a vertex is optimal.
Now assume \(k > 0\).
If \(k < 0\), reflect the coordinate plane about \(y = 0\).
Geometrically, it is easy to see (\cref{fig:computation}) that the coordinate of interest can be
expressed as
\[x = x' - \sqrt{\varepsilon^2 - y'^2} = \frac{y' - b}{k} - \sqrt{\varepsilon^2 - y'^2}\,.\]
We want to maximise \(x\) by picking the appropriate \(y'\).
We take the derivative:
\[\frac{\mathrm{d}x}{\mathrm{d}y'} = \frac{1}{k} + \frac{y'}{\sqrt{\varepsilon^2 - y'^2}}\,.\]
We can equate it to \(0\) to find the critical point of the function.
Simplifying, we find
\[y'_0 = -\frac{\varepsilon}{\sqrt{k^2 + 1}}\,.\]
We can check that for \(y' < y'_0\), the value of the derivative is negative, and for \(y' > y'_0\)
it is positive, so at \(y' = y'_0\) we achieve a local minimum.
There are no other critical points.
Therefore, to maximise \(x\), we want to move as far as possible in either direction, away from
the local minimum.
Since we are limited to the line segment \(uv\), the maximum is clearly achieved at one of the
segment endpoints.
\end{proof}

\section{Shortcut Testing: All Points}\label{sec:shortcut}
In the previous \lcnamecref{sec:intermediate}, we have covered testing a shortcut, given that the
first and the last points are fixed.
Here we remove that restriction.
\begin{problem}\label{prob:shortcut}
Given an uncertain curve \(\U = \langle U_1, \dots, U_n\rangle\) on \(n \in \N\), \(n \geq 3\)
uncertain points in \(\Rt\), check if the largest Hausdorff or Fr\'echet distance between \(\U\) and
its one-segment simplification is below a threshold \(\varepsilon \in \R^{> 0}\) for any
realisation, i.e.\@ for \(\delta \coloneqq \hs\) or \(\delta \coloneqq \fr\), verify
\(\max_{\pi \Subset \U} \delta(\pi, p_1p_n) \leq \varepsilon\).
\end{problem}

We first show how this can be done for indecisive points, both for \(\delta \coloneqq \hs\) and
\(\delta \coloneqq \fr\).
\begin{lemma}\label{lem:indecisive_main}
Given \(n, k \in \N^{> 0}\), \(n \geq 3\), and \(\delta \coloneqq \hs\) or \(\delta \coloneqq \fr\),
for any indecisive curve \(\U = \langle U_1, \dots, U_n\rangle\) with
\(U_i = \{p_i^1, \dots, p_i^k\}\) for all \(i \in [n]\) and \(p_i^j \in \Rt\) for all \(i \in [n]\),
\(j \in [k]\), we have
\[\max_{\pi \Subset \U} \delta(\pi, \langle \pi(1), \pi(n)\rangle) =
\multiadjustlimits{\max_{a \in [k]}, \max_{b \in [k]}, \max_{\sigma \Subset \U,
\sigma(1) \equiv p_1^a, \sigma(n) \equiv p_n^b}} \delta(\sigma, p_1^ap_n^b)\,.\]
\end{lemma}
\begin{proof}
We can derive
\begin{align*}
&\max_{\pi \Subset \U} \delta(\pi, \langle \pi(1), \pi(n)\rangle)\\
&\reason{Def.~\(\Subset\)}\\
&= \max_{p_1 \in U_1, \dots, p_n \in U_n} \delta(\langle p_1, \dots, p_n\rangle, p_1p_n)\\
&= \max_{p_1 \in U_1} \max_{p_n \in U_n} \max_{p_2 \in U_2, \dots, p_{n - 1} \in U_{n - 1}}
\delta(\langle p_1, \dots, p_n\rangle, p_1p_n)\\
&\reason{Def.~\(\Subset\)}\\
&= \multiadjustlimits{\max_{p_1 \in U_1}, \max_{p_n \in U_n},
\max_{\sigma \Subset \U, \sigma(1) \equiv p_1, \sigma(n) \equiv p_n}} \delta(\sigma, p_1p_n)\\
&= \multiadjustlimits{\max_{a \in [k]}, \max_{b \in [k]},
\max_{\sigma \Subset \U, \sigma(1) \equiv p_1^a, \sigma(n) \equiv p_n^b}}
\delta(\sigma, p_1^ap_n^b)\,,
\end{align*}
as was to be shown.
\end{proof}
That is to say, for either Hausdorff or Fr\'echet distance we can simply test the shortcut using the
corresponding procedure from \cref{lem:hausdorff_ind} or \cref{lem:frechet_indecisive}, and do so
for each combination of the start and end points.
We can then test an indecisive shortcut of length \(n\) overall in time
\(\bO(k^2 \cdot nk) = \bO(nk^3)\).

We now proceed to show the approach for disks and polygonal closed convex sets.
The procedure is the same for the Hausdorff and the Fr\'echet distance, but differs between disks
and PCCSs, since disks have some convenient special properties.

\subsection{Disks}
We start by stating some useful observations.
\begin{observation}\label{obs:disks}
Suppose we are given two non-degenerate disks \(D_1 \coloneqq D(p_1, r_1)\) and
\(D_2 \coloneqq D(p_2, r_2)\) with \(D_1 \nsubseteq D_2\) and \(D_2 \nsubseteq D_1\).
We make the following observations.
(See \cref{fig:disks}.)
\begin{itemize}
\item There are exactly two \emph{outer tangents} to the disks, and the convex hull of
\(D_1 \cup D_2\) consists of an arc from \(D_1\), an arc from \(D_2\), and the outer tangents.
\item Assume the lines of the outer tangents intersect.
When viewed from the intersection point, the order in which the tangents touch the disks is the
same, i.e.\@ either both first touch \(D_1\) and then \(D_2\), or the other way around.
If the lines are parallel, the same statement holds when viewed from points on the tangent lines
at infinity.
\end{itemize}
\end{observation}
To see that the second observation is true, note that the distance from the intersection point to
the tangent points of a disk is the same for both tangent lines.
These observations mean that we can restrict our attention to the area bounded by the outer
tangents and define an ordering in the resulting strip.

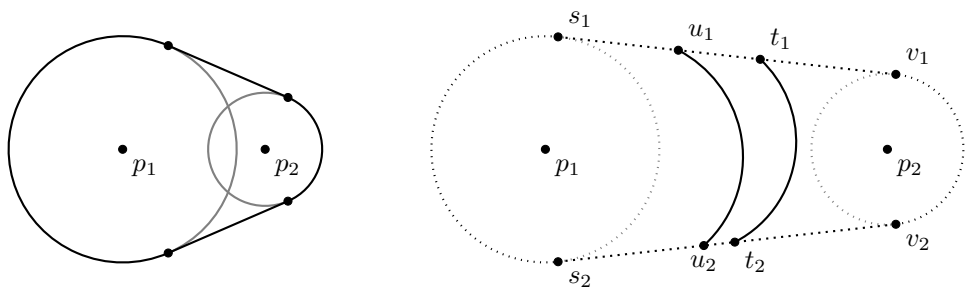
\begin{figure}[tb]
\begin{minipage}{.35\linewidth}
\centering
\begin{tikzpicture}[scale=.75]
\tkzInit[xmin=-2.5,xmax=4,ymin=-2.5,ymax=2.5]
\tkzClip
\tkzDefPoint(0, 0){p_1}
\tkzDefPoint(2.5, 0){p_2}

\tkzExtSimilitudeCenter(p_1, 2)(p_2, 1)
\tkzGetPoint{O}
\tkzDefTangent[from with R=O](p_1, 2 cm)
\tkzGetPoints{A}{B}
\tkzDefTangent[from with R=O](p_2, 1 cm)
\tkzGetPoints{A'}{B'}

\tkzDrawArc[color=black,thick](p_1,A)(B)
\tkzDrawArc[color=black,thick](p_2,B')(A')
\tkzDrawArc[color=gray,thick](p_1,B)(A)
\tkzDrawArc[color=gray,thick](p_2,A')(B')
\tkzDrawSegments[thick](A,A' B,B')
\tkzDrawPoints(p_1, p_2, A, B, A', B')
\tkzLabelPoints(p_1, p_2)
\end{tikzpicture}
\end{minipage}%
\begin{minipage}{.65\linewidth}
\centering
\begin{tikzpicture}[scale=1]
\tkzInit[xmin=-2,xmax=6,ymin=-2,ymax=2]
\tkzClip
\tkzDefPoint(0, 0){p_1}
\tkzDefPoint(4.5, 0){p_2}
\tkzDefPoint(1.8, .1){p_i}
\tkzDefPoint(1, -.1){p_j}

\tkzExtSimilitudeCenter(p_1, 1.5)(p_2, 1)
\tkzGetPoint{O}
\tkzDefTangent[from with R=O](p_1, 1.5 cm)
\tkzGetPoints{s_1}{s_2}
\tkzDefTangent[from with R=O](p_2, 1 cm)
\tkzGetPoints{v_1}{v_2}

\tkzInterLC[R](s_1,v_1)(p_i,1.5cm)
\tkzGetPoints{dummy1}{t_1}
\tkzInterLC[R](s_2,v_2)(p_i,1.5cm)
\tkzGetPoints{t_2}{dummy2}

\tkzInterLC[R](s_1,v_1)(p_j,1.6cm)
\tkzGetPoints{dummy1}{u_1}
\tkzInterLC[R](s_2,v_2)(p_j,1.6cm)
\tkzGetPoints{u_2}{dummy2}

\tkzDrawArc[color=black,thick](p_i,t_2)(t_1)
\tkzDrawArc[color=black,thick](p_j,u_2)(u_1)
\tkzDrawArc[color=black,dotted,thick](p_1,s_1)(s_2)
\tkzDrawArc[color=black,dotted,thick](p_2,v_2)(v_1)
\tkzDrawArc[color=gray,dotted,thick](p_1,s_2)(s_1)
\tkzDrawArc[color=gray,dotted,thick](p_2,v_1)(v_2)
\tkzDrawSegments[thick,dotted](s_1,v_1 s_2,v_2)
\tkzDrawPoints(p_1, p_2, s_1, s_2, v_1, v_2, t_1, t_2, u_1, u_2)
\tkzLabelPoints(p_1,p_2,s_2,t_2,v_2)
\tkzLabelPoints[above right](s_1,u_1,t_1,v_1)
\tkzLabelPoints[below](u_2)
\end{tikzpicture}
\end{minipage}
\caption{Left: Illustration for \cref{obs:disks}.
The convex hull of the disks is highlighted in black.
The order in which the outer tangents touch the disks is the same.
Right: Illustration for \cref{def:order}.
Here \(O_1\) (\(t_1\) to \(t_2\)) is to the right of \(O_2\) (\(u_1\) to \(u_2\)).}
\label{fig:disks}
\end{figure}

\begin{definition}\label{def:order}
Given two distinct non-degenerate disks \(D_1 \coloneqq D(p_1, r_1)\) and
\(D_2 \coloneqq D(p_2, r_2)\), consider a strip defined by the lines that form the outer tangents to
the disks.
Assume we have two circular arcs \(O_1, O_2\) that intersect both tangents and lie inside the strip.
Define \(s_1\) and \(v_1\) to be the points where one of the tangents touches \(D_1\) and \(D_2\),
respectively, and let \(t_1\) and \(u_1\) be the points where \(O_1\) and \(O_2\) intersect that
tangent, respectively.
Define the order on the tangents from \(D_1\) to \(D_2\), so \(s_1 \prec v_1\).
Define points \(s_2\), \(t_2\), \(u_2\), \(v_2\) similarly for the other tangent.
We say that \(O_2\) is \emph{to the right} of \(O_1\) if either \(t_i = u_i\) for \(i \in \{1, 2\}\)
and the radius of \(O_1\) is larger than that of \(O_2\); or if otherwise \(t_i \preccurlyeq u_i\)
for \(i \in \{1, 2\}\) and \(O_1\) and \(O_2\) do not properly intersect.
We say that \(O_2\) is \emph{to the left} of \(O_1\) if either \(t_i = u_i\) for \(i \in \{1, 2\}\)
and the radius of \(O_1\) is smaller than that of \(O_2\); or if otherwise \(u_i \preccurlyeq t_i\)
for \(i \in \{1, 2\}\) and \(O_1\) and \(O_2\) do not properly intersect.
(See \cref{fig:disks} for a visual interpretation.)
\end{definition}

We are now ready to state the main result for the Hausdorff distance.
\begin{lemma}\label{lem:hausdorff_disks_main}
Given \(n \in \N^{> 0}\), \(n \geq 3\), for any imprecise curve modelled with disks
\(\U = \langle U_1, \dots, U_n\rangle\) with \(U_i = D(c_i, r_i)\) for all \(i \in [n]\) and
\(c_i \in \Rt\), \(r_i \in \R^{\geq 0}\) for all \(i \in [n]\), and assuming \(U_1 \neq U_n\), we
have
\[\max_{\pi \Subset \U} \hs(\pi, \langle\pi(1), \pi(n)\rangle) \leq \varepsilon\]
if and only if both of the following are true:
\vspace{\abovedisplayskip}
\begin{itemize}
\item \(\displaystyle\max\Big\{\max_{\pi \Subset \U, \pi(1) \equiv s, \pi(n) \equiv t} \hs(\pi, st),
\max_{\pi \Subset \U, \pi(1) \equiv u, \pi(n) \equiv v} \hs(\pi, uv)\Big\}
\leq \varepsilon\,,\)\hfil\\[\belowdisplayskip]
where \(s, u \in U_1\), \(t, v \in U_n\), and \(st\) and \(uv\) are the outer tangents to
\(U_1 \cup U_n\); and
\item for each \(i \in \{2, \dots, n - 1\}\), the right arc of the disk
\(D(c_i, \varepsilon - r_i)\) bounded by the intersection points with the tangent lines is to the
right of the right arc of \(U_1\) and the left arc of the disk \(D(c_i, \varepsilon - r_i)\) is to
the left of the left arc of \(U_n\).
\end{itemize}
\end{lemma}
\begin{proof}
Assume the right side of the \lcnamecref{lem:hausdorff_disks_main} statement holds.
First of all, as we have
\(\max_{\pi \Subset \U, \pi(1) \equiv s, \pi(n) \equiv t} \hs(\pi, st) \leq \varepsilon\),
\cref{lem:hausdorff_disks} shows that for all \(i \in \{2, \dots, n - 1\}\), we have
\(d(c_i, st) + r_i \leq \varepsilon\), or \(d(c_i, st) \leq \varepsilon - r_i\), so
\(st\) stabs each disk \(D(c_i, \varepsilon - r_i)\).
We can draw a similar conclusion for \(uv\).
Therefore, each disk \(D(c_i, \varepsilon - r_i)\) crosses the entire strip bounded by the tangent
lines, with the intersection points splitting it into the left and the right circular arcs.
We can thus apply \cref{def:order} to these arcs, as stated in the lemma.

First suppose that the disks \(U_1\) and \(U_n\) do not intersect.
Then for any line segment from \(U_1\) to \(U_n\) and any disk
\(D' \coloneqq D(c_i, \varepsilon - r_i)\), we exit \(D'\) after exiting \(U_1\) and enter \(D'\)
before entering \(U_n\).
Hence, for any line \(pq\) with \(p \in U_1\) and \(q \in U_n\) and any
\(i \in \{2, \dots, n - 1\}\), we can find a point \(w \in pq \cap D'\); this means, as stated in
\cref{lem:aux_disk_segm}, that indeed \(\max_{w' \in U_i} d(w', pq) \leq \varepsilon\).
As this holds for all disks and any choice of \(p\) and \(q\), we conclude that
\(\max_{\pi \Subset \U} \hs(\pi, \langle\pi(1), \pi(n)\rangle) \leq \varepsilon\).

Now assume that the disks \(U_1\) and \(U_n\) intersect.
If we consider the line segments \(pq\) with \(p \in U_1\), \(q \in U_n\), we end up in the previous
case if either \(p \notin U_1 \cap U_n\) or \(q \notin U_1 \cap U_n\).
So assume that the segment \(pq\) lies entirely in the intersection \(U_1 \cap U_n\).
However, it can be seen that for each disk \(D' \coloneqq D(c_i, \varepsilon - r_i)\), the left
boundary of the intersection is to the right of the left boundary of the disk, and the right
boundary of the intersection is to the left of the right boundary of the disk; hence,
\(pq \subset U_1 \cap U_n \subseteq D'\).
Therefore, we have \(\max_{w' \in U_i} d(w', pq) \leq \varepsilon\), and so also in this case
\(\max_{\pi \Subset \U} \hs(\pi, \langle\pi(1), \pi(n)\rangle) \leq \varepsilon\).

We now assume that the right side of the lemma statement is false and show that then
\(\max_{\pi \Subset \U} \hs(\pi, \langle\pi(1), \pi(n)\rangle) > \varepsilon\).
If \(\max_{\pi \Subset \U, \pi(1) \equiv s, \pi(n) \equiv t} \hs(\pi, st) > \varepsilon\), then
immediately \(\max_{\pi \Subset \U} \hs(\pi, \langle\pi(1), \pi(n)\rangle) > \varepsilon\).
Same holds for \(uv\).
So, assume those statements hold; then it must be that for at least one intermediate disk the arcs
do not lie to the left or to the right of the arcs of the respective disks.
Assume this is disk \(i\), so the disk \(D' \coloneqq D(c_i, \varepsilon - r_i)\).
W.l.o.g.\@ assume that the right arc of the disk does not lie entirely to the right of the right arc
of \(U_1\).
The argument for the left arc w.r.t.\@ \(U_n\) is symmetric.

There must be at least one point \(p'\) on the right arc of \(U_1\) that lies outside of \(D'\).
Assume for now that \(U_1\) and \(U_n\) are disjoint.
Then a line segment \(p'q\) for any \(q \in U_n\) does not stab \(D'\), so
\(\max_{w' \in U_i} d(w', pq) > \varepsilon\), and so
\(\max_{\pi \Subset \U} \hs(\pi, \langle\pi(1), \pi(n)\rangle) > \varepsilon\).
If \(U_1\) and \(U_n\) intersect, then either \(p'\) is outside of the intersection
and of \(D'\) and there is a point \(q \in U_n\) such that \(p'q\) does not stab \(D'\); or we can
pick the degenerate line segment \(p'p'\), as \(p' \in U_1 \cap U_n\), and so \(p'p'\) also does not
stab \(D'\).
In either case, we conclude that
\(\max_{\pi \Subset \U} \hs(\pi, \langle\pi(1), \pi(n)\rangle) > \varepsilon\).
\end{proof}
It is also worth noting that the case of \(U_1 = U_n\) is similar to how we treat the intersection
\(U_1 \cap U_n\) above; however, our \lcnamecref{def:order} for the ordering between two disks does
not apply.
So, if \(U_1 = U_n\), then \(\max_{\pi \Subset \U} \hs(\pi, \langle\pi(1), \pi(n)\rangle) \leq \varepsilon\)
if and only if \(U_1 \subseteq D(c_i, \varepsilon - r_i)\) for all \(i \in \{2, \dots, n - 1\}\).

Similarly, we state the following for the Fr\'echet distance.
\begin{lemma}\label{lem:frechet_disks_main}
Given \(n \in \N^{> 0}\), \(n \geq 3\), for any imprecise curve modelled with disks
\(\U = \langle U_1, \dots, U_n\rangle\) with \(U_i = D(c_i, r_i)\) for all \(i \in [n]\) and
\(c_i \in \Rt\), \(r_i \in \R^{\geq 0}\) for all \(i \in [n]\), and assuming \(U_1 \neq U_n\), we
have
\[\max_{\pi \Subset \U} \fr(\pi, \langle\pi(1), \pi(n)\rangle) \leq \varepsilon\]
if and only if both of the following are true:
\vspace{\abovedisplayskip}
\begin{itemize}
\item \(\displaystyle\max\Big\{\max_{\pi \Subset \U, \pi(1) \equiv s, \pi(n) \equiv t} \fr(\pi, st),
\max_{\pi \Subset \U, \pi(1) \equiv u, \pi(n) \equiv v} \fr(\pi, uv)\Big\}
\leq \varepsilon\,,\)\hfil\\[\belowdisplayskip]
where \(s, u \in U_1\), \(t, v \in U_n\), and \(st\) and \(uv\) are the outer tangents to
\(U_1 \cup U_n\); and
\item for each \(i \in \{2, \dots, n - 1\}\), the right arc of the disk
\(D(c_i, \varepsilon - r_i)\) bounded by the intersection points with the tangent lines is to the
right of the right arc of \(U_1\) and the left arc of the disk \(D(c_i, \varepsilon - r_i)\) is to
the left of the left arc of \(U_n\).
\end{itemize}
\end{lemma}
\begin{proof}
First assume that \(\max_{\pi \Subset \U} \fr(\pi, \langle\pi(1), \pi(n)\rangle) \leq \varepsilon\).
As \(\fr(\pi, \sigma) \leq \hs(\pi, \sigma)\) for any curves \(\pi\), \(\sigma\), this also means
that \(\max_{\pi \Subset \U} \hs(\pi, \langle\pi(1), \pi(n)\rangle) \leq \varepsilon\).
Furthermore, immediately we get that
\(\max_{\pi \Subset \U, \pi(1) \equiv s, \pi(n) \equiv t} \fr(\pi, st) \leq \varepsilon\), and the
same for \(uv\).
Together with \cref{lem:hausdorff_disks_main}, this yields the right side of the lemma.

Now assume that the right side holds.
As in \cref{lem:hausdorff_disks_main}, we know that the disks cross the entire strip and that
\cref{def:order} applies.
It remains to show that for any line segment \(pq\) with \(p \in U_1\), \(q \in U_n\), there is a
valid alignment that maintains the correct ordering and bottleneck distance, assuming it exists for
every realisation for \(st\) and \(uv\).
Consider a valid alignment established for \(st\) and \(uv\), so the sequence of points \(a_i\) on
\(st\) and \(b_i\) on \(uv\) that are mapped to \(U_i\).
As we showed in \cref{lem:aux_disk_point}, we can always find such points for each individual
\(U_i\), and as we know that the Fr\'echet distance is below the threshold for \(st\) and \(uv\),
there is such a valid alignment, i.e.\@ we know that \(a_i \preccurlyeq a_{i + 1}\) and
\(b_i \preccurlyeq b_{i + 1}\) for all \(i \in [n - 1]\).

First suppose that the disks \(U_1\) and \(U_n\) do not intersect.
Consider the region \(R\) bounded by the outer tangents and the disk arcs that are not part of the
convex hull of \(U_1 \cup U_n\).
We connect, for each \(i \in \{2, \dots, n - 1\}\), \(a_i\) to \(b_i\) with a geodesic shortest path
in \(R\).
We claim that for any line segment \(pq\) defined above, the intersection points of the shortest
paths with the segment give a valid alignment, yielding
\(\max_{\pi \Subset \U, \pi(1) \equiv p, \pi(n) \equiv q} \fr(\pi, pq) \leq \varepsilon\).
As the choice of \(pq\) was arbitrary, this will complete the proof.

To show that the alignment is valid, we need to show that the order is correct and that the distances
fall below the threshold.
First consider the case where the geodesic shortest path for point \(i\) does not touch the boundary
formed by arcs of region \(R\).
In this case, it is simply a line segment \(a_ib_i\).
Note that by definition \(a_i, b_i \in D(c_i, \varepsilon - r_i)\); as disks are convex, also
\(a_ib_i \subset D(c_i, \varepsilon - r_i)\); thus, the intersection point \(p'_i\) of \(pq\) with
\(a_ib_i\) is in \(D(c_i, \varepsilon - r_i)\), so by \cref{lem:aux_disk_point},
\(\max_{w \in U_i} \lVert p'_i - w\rVert \leq \varepsilon\).
Furthermore, note that \(a_i \preccurlyeq a_{i + 1}\) and \(b_i \preccurlyeq b_{i + 1}\);
thus, the line segments \(a_ib_i\) and \(a_{i + 1}b_{i + 1}\) cannot cross, so also
\(p'_i \preccurlyeq p'_{i + 1}\).

Now w.l.o.g.\@ consider the case where the geodesic shortest path for point \(i\) touches the arc of
\(U_1\).
The geodesic shortest paths do not cross: on the path from \(a_i\) (or \(b_i\)) to the arc they
form a tangent to the arc, thus for \(a_i \preccurlyeq a_{i + 1}\) the tangent point for \(a_i\)
comes before that of \(a_{i + 1}\) when going along the arc from \(s\) to \(u\).
So, just as in the previous case, these line segments cannot cross.
Having reached the arc, both shortest paths will follow it, as otherwise the path would not be
a shortest path; thus, the arcs do not cross, either.
Finally, a path from the previous case does not touch any path that touches the arc boundary of
\(R\) by definition.
Finally, note that the condition that we have established on the right arcs of disks being to the
right of the right arc of \(U_1\) (and symmetric for the left arcs and \(U_n\)) means that the
geodesic shortest paths that touch the arc boundary of \(R\) stay within the respective disks
\(D(c_i, \varepsilon - r_i)\).
Thus, we have established that for all \(i\) we have \(p'_i \preccurlyeq p'_{i + 1}\) and
\(\max_{w \in U_i} \lVert p'_i - w\rVert \leq \varepsilon\), concluding the proof for disjoint
\(U_1\) and \(U_n\).

Finally, consider the case where \(U_1\) intersects \(U_n\).
Above we used geodesic paths within the region \(R\).
However, when \(U_1\) intersects \(U_n\), \(R\) consists of two disconnected regions.
Observe that one region contains \(a_i\) and the other contains \(b_i\).
To connect \(a_i\) with \(b_i\) we use the geodesic from \(a_i\) to the intersection point of the
two inner boundaries of \(U_1\) and \(U_n\) that is in the same region of \(R\), the geodesic from
\(b_i\) to the other intersection point of the inner boundaries, and join these two by a line
segment between the intersection points.
Any line segment from a point in \(U_1\) to a point in \(U_n\) crosses these paths in order, just
like in the previous case.
If the line segment goes through the intersection, note that any point in the intersection is close
enough to all the intermediate objects, as the intersection is the subset of each disk.
So, any point in the intersection can be chosen to establish the trivially in-order alignment to all
the intermediate objects.
\end{proof}
Again, in the case that \(U_1 = U_n\), we can see that
\(\max_{\pi \Subset \U} \fr(\pi, \langle\pi(1), \pi(n)\rangle) \leq \varepsilon\)
if and only if \(U_1 \subseteq D(c_i, \varepsilon - r_i)\) for all \(i \in \{2, \dots, n - 1\}\).

\subsection{Non-intersecting PCCSs}\label{sec:non_int_pccs}
Suppose the regions are modelled by convex polygons.
Consider first the case where the interiors of \(U_1\) and \(U_n\) do not intersect, so at most
they share a boundary segment.
\begin{observation}\label{obs:pccs}
Given an uncertain curve modelled by convex polygons \(\U = \langle U_1, \dots, U_n\rangle\) with
the interiors of \(U_1\) and \(U_n\) not intersecting, note:
\begin{itemize}
\item There are two \emph{outer tangents} to the polygons \(U_1\) and \(U_n\), and the convex hull
of \(U_1 \cup U_2\) consists of a convex chain from \(U_1\), a convex chain from \(U_n\), and the
outer tangents.
\item Let \(C_i\) be the convex chain from \(U_i\) that is not part of the convex hull for
\(i \in \{1, n\}\).
Then for \(\delta \coloneqq \hs\) or \(\delta \coloneqq \fr\),
\[\max_{\pi \Subset \U} \delta(\pi, \langle\pi(1), \pi(n)\rangle) \leq \varepsilon \iff
\max_{\pi \Subset \U, \pi(1) \in C_1, \pi(n) \in C_n}
\delta(\pi, \langle\pi(1), \pi(n)\rangle) \leq \varepsilon\,.\]
\end{itemize}
\end{observation}
To see that the second observation is true, note that one direction is trivial.
In the other direction, note that any line segment \(pq\) with \(p \in U_1\), \(q \in U_n\) crosses
both \(C_1\) and \(C_n\), say, at \(p' \in C_1\) and \(q' \in C_n\).
We know that there is a valid alignment for \(p'q'\), both for the Hausdorff and the Fr\'echet
distance; we can then use this alignment for \(pq\).
See \cref{fig:polygons}.

\begin{figure}[tb]
\begin{minipage}{.5\linewidth}
\centering
\begin{tikzpicture}[scale=2]
\tkzDefPoints{0/0/p,2/0/q}
\tkzDefPoints{-.5/.1/a_1,.1/.5/a_2,.4/-.3/a_3,.2/-.5/a_4,-.3/-.2/a_5}
\tkzDefPoints{1.5/-.2/b_1,1.5/.2/b_2,1.9/.5/b_3,2.4/-.3/b_4,2.1/-.5/b_5,1.6/-.3/b_6}

\tkzInterLL(p,q)(a_2,a_3)
\tkzGetPoint{p'}
\tkzInterLL(p,q)(b_1,b_2)
\tkzGetPoint{q'}

\tkzDrawPolySeg[color=gray](a_4,a_5,a_1,a_2)
\tkzDrawPolySeg[color=gray](b_3,b_4,b_5)
\tkzDrawSegments[color=gray](a_2,b_3 a_4,b_5)
\tkzDrawPolySeg[dashed](a_2,a_3,a_4)
\tkzDrawPolySeg[dashed](b_5,b_6,b_1,b_2,b_3)
\tkzDrawSegment(p,q)

\tkzDrawPoints(p,q,p',q')
\tkzLabelPoints[above right](p')
\tkzLabelPoints[above left](q')
\tkzLabelPoints[below](p,q)
\tkzLabelPoint[right](a_3){\(C_1\)}
\tkzLabelPoint[left](b_6){\(C_n\)}
\tkzLabelSegment[xshift=2pt](a_1,a_2){\(U_1\)}
\tkzLabelSegment[xshift=-2pt](b_3,b_4){\(U_n\)}
\end{tikzpicture}
\end{minipage}%
\begin{minipage}{.5\linewidth}
\centering
\begin{tikzpicture}[scale=2]
\tkzDefPoints{-.5/.1/a_1,.1/.5/a_2,.4/-.3/a_3,.2/-.5/a_4,-.3/-.2/a_5}
\tkzDefPoints{1.5/-.2/b_1,1.5/.2/b_2,1.9/.5/b_3,2.4/-.3/b_4,2.1/-.5/b_5,1.6/-.3/b_6}
\tkzDefPoint(.5, 0){R}

\tkzDrawPolySeg[dotted](a_4,a_5,a_1,a_2)
\tkzDrawPolySeg[dotted](b_3,b_4,b_5)
\tkzDrawPolygon(a_2,a_3,a_4,b_5,b_6,b_1,b_2,b_3)
\tkzDrawSegments[dashed](a_2,b_2 a_3,b_2 a_3,b_1 a_3,b_6 a_4,b_6)
\tkzLabelPoints[above right](R)
\end{tikzpicture}
\end{minipage}
\caption{Left: Illustration for \cref{obs:pccs}.
The convex hull of the disks is shown in grey.
The dotted chains are \(C_1\) and \(C_n\).
Any line segment \(pq\) with \(p \in U_1\) and \(q \in U_n\) crosses \(C_1\) and \(C_n\).
Right: Illustration for the procedure.
The region \(R\) is triangulated.}
\label{fig:polygons}
\end{figure}
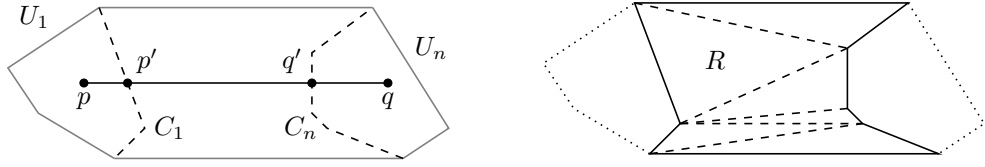

We claim that we can use the following procedure to check
\(\max_{\pi \Subset \U} \hs(\pi, \langle\pi(1), \pi(n)\rangle) \leq \varepsilon\).
\begin{enumerate}
\item Triangulate the region \(R\) bounded by two convex chains \(C_1\) and \(C_n\) and the outer
tangents.
\item For each line segment \(st\) of the triangulation with \(s \in C_1\), \(t \in C_n\), and for
either \(\delta \coloneqq \hs\) or \(\delta \coloneqq \fr\), check that
\(\max_{\pi \Subset \U, \pi(1) \equiv s, \pi(n) \equiv t} \delta(\pi, st) \leq \varepsilon\).
\end{enumerate}

First of all, observe that we can compute a triangulation, and that every triangle has two points
from one convex chain and one point from the other chain (see \cref{fig:polygons}).
If all three points were from the same chain, then the triangle would lie outside of \(R\).
Now consider some line segment \(pq\) with \(p \in C_1\), \(q \in C_n\).
To complete the argument, it remains to show that the checks in step~2 mean that also
\(\max_{\pi \Subset \U, \pi(1) \equiv p, \pi(n) \equiv q} \delta(\pi, pq) \leq \varepsilon\).
Observe that the triangles span across the region \(R\), so when going from one tangent to the other
within \(R\) we cross all the triangles.
Therefore, we can order them, in the order of occurrence on such a path, from \(1\) to \(k\).
Denote the alignment established on line \(j \in [k]\) with the sequence of \(a_i^j\), for
\(i \in [n]\); this alignment can be established both for \(\delta \coloneqq \hs\) and
\(\delta \coloneqq \fr\).
We can then establish polygonal curves \(A_i \coloneqq \langle a_i^1, \dots, a_i^k\rangle\);
clearly, they all stay within \(R\).
We claim that for any line segment \(pq\) defined above, it is possible to establish a valid
alignment from intersection points of \(pq\) and \(A_i\).
We do this separately for the Fr\'echet and the Hausdorff distance.

\begin{lemma}\label{lem:hausdorff_pccs_main}
Given a set of curves \(A \coloneqq \{A_2, \dots, A_{n - 1}\}\) in \(R\) described above for
\(\delta \coloneqq \hs\) and a line segment \(pq\) with \(p \in C_1\), \(q \in C_n\), we have
\(\max_{\pi \Subset \U, \pi(1) \equiv p, \pi(n) \equiv q} \hs(\pi, pq) \leq \varepsilon\).
\end{lemma}
\begin{proof}
Note that \(pq\) crosses each \(A_i\) at least once.
We can take any one crossing for each \(i\) and establish the alignment.
Consider such a crossing point \(p'_i\).
It falls in some triangle bounded by a segment from either \(C_1\) or \(C_n\) and two line segments
that contain points \(a_i^j\) and \(a_i^{j + 1}\) for some \(j \in [k]\).
We know, using \cref{lem:hausdorff_pccs}, that
\(\max_{w \in U_i} \lVert a_i^j - w\rVert \leq \varepsilon\) and
\(\max_{w \in U_i} \lVert a_i^{j + 1} - w\rVert \leq \varepsilon\).
Consider any point \(w' \in U_i\).
Then, using \cref{lem:distance_segm} with \(c \coloneqq d \coloneqq w'\), we find that
\(\lVert w' - p'_i\rVert \leq \varepsilon\).
Therefore, also \(\max_{w \in U_i} \lVert p'_i - w\rVert \leq \varepsilon\); using
\cref{lem:hausdorff_pccs}, we conclude that indeed
\(\max_{\pi \Subset \U, \pi(1) \equiv p, \pi(n) \equiv q} \hs(\pi, pq) \leq \varepsilon\).
\end{proof}

For the Fr\'echet distance, we can use the same argument to show closeness; however, we need more
care to establish the correct order for the alignment to be valid.
\begin{lemma}\label{lem:frechet_pccs_main}
Given a set of curves \(A \coloneqq \{A_2, \dots, A_{n - 1}\}\) in \(R\) described above
for \(\delta \coloneqq \fr\) and a line segment \(pq\) with \(p \in C_1\), \(q \in C_n\), we have
\(\max_{\pi \Subset \U, \pi(1) \equiv p, \pi(n) \equiv q} \fr(\pi, pq) \leq \varepsilon\).
\end{lemma}
\begin{proof}
Compared to \cref{lem:hausdorff_pccs_main}, instead of taking any intersection point of \(pq\) with
each \(A_i\), we take the \emph{last} intersection point.

We need to show, first of all, that curves \(A_i\) and \(A_{i + 1}\) do not cross for any
\(i \in [n - 1]\).
Note that each curve \(A_i\) crosses each triangle once, so it suffices to show that a segment
\(a_i^ja_i^{j + 1}\) does not cross \(a_{i + 1}^ja_{i + 1}^{j + 1}\).
Indeed, as \(a_i^j \preccurlyeq a_{i + 1}^j\) and \(a_i^{j + 1} \preccurlyeq a_{i + 1}^{j + 1}\),
these line segments cannot cross.

Now consider, for each \(i \in \{2, \dots, n - 1\}\). the polygon \(P_i\) bounded by \(C_1\),
\(A_i\), and the corresponding segments of the outer tangents.
With the previous statement, it is easy to see that
\(P_2 \subseteq P_3 \subseteq \dots \subseteq P_{n - 1}\).
Assume this is not the case, so some \(P_i \not\subseteq P_{i + 1}\).
Then there is a point \(z \in P_i\), but \(z \notin P_{i + 1}\).
The point \(z\) falls into some triangle with lines \(j\) and \(j + 1\).
In this triangle, it means that \(z\) is between \(C_1\) and \(a_i^ja_i^{j + 1}\), but not between
\(C_1\) and \(a_{i + 1}^ja_{i + 1}^{j + 1}\).
However, as these segments do not cross, this would imply that \(a_{i + 1}^j \prec a_i^j\), but then
the check in step~2 would not pass for line \(j\).

Consider the points at which the line segment \(pq\) leaves the polygons \(P_i\) for the last time.
From the definition it is obvious that \(p \in P_i\) for all \(i \in \{2, \dots, n - 1\}\), so this
is well-defined.
Clearly, due to the subset relationship, the order of such points \(p'_i\) is correct, i.e.\@
\(p'_i \preccurlyeq p'_{i + 1}\).
Furthermore, each such \(p'_i \in A_i\), so using the arguments of \cref{lem:hausdorff_pccs_main}
we can show that also the distances are below \(\varepsilon\).
Thus, we conclude that indeed
\(\max_{\pi \Subset \U, \pi(1) \equiv p, \pi(n) \equiv q} \fr(\pi, pq) \leq \varepsilon\).
\end{proof}

The proofs of \cref{lem:hausdorff_pccs_main,lem:frechet_pccs_main} show us how to solve the problem
for two convex polygons with non-intersecting interiors.
We can also use them directly for the case of line segments that do not intersect except at
endpoints.
\begin{corollary}\label{cor:line_segments}
Given \(n \in \N^{> 0}\), \(n \geq 3\), for any imprecise curve modelled with line segments
\(\U = \langle U_1, \dots, U_n\rangle\) with \(U_i = p_i^1p_i^2 \subset \Rt\) for all \(i \in [n]\),
given a threshold \(\varepsilon \in \R^{> 0}\), and given that
\(U_1 \cap U_n \subset \{p_1^1, p_1^2\}\), and assuming that the triangles \(p_1^1 p_n^1 p_1^2\) and
\(p_1^2 p_n^1 p_n^2\) form a triangulation of the convex hull of \(U_1 \cup U_n\), we have
\[\max_{\pi \Subset \U} \delta(\pi, \langle\pi(1), \pi(n)\rangle) \leq \varepsilon\]
if and only if
\begin{align*}
\max\big\{
&\max_{\pi \Subset \U, \pi(1) \equiv p_1^1, \pi(n) \equiv p_n^1} \delta(\pi, p_1^1p_n^1)\,,\\
&\max_{\pi \Subset \U, \pi(1) \equiv p_1^2, \pi(n) \equiv p_n^1} \delta(\pi, p_1^2p_n^1)\,,\\
&\max_{\pi \Subset \U, \pi(1) \equiv p_1^2, \pi(n) \equiv p_n^2} \delta(\pi, p_1^2p_n^2)\big\}
\leq \varepsilon\,.
\end{align*}
\end{corollary}
We should note that in this particular case it is not necessary to use a triangulation, so we can
get rid of the second term; also in the previous proofs a convex partition could work instead, but a
triangulation is easier to define.

\subsection{Intersecting PCCSs}
We proceed to discuss the situation where the interiors of \(U_1\) and \(U_n\) intersect, or
where line segments \(U_1\) and \(U_n\) cross.
The argument is the same for both \(\delta \coloneqq \hs\) and \(\delta \coloneqq \fr\), but it is
easier to treat line segments and convex polygons separately.

\paragraph*{Line segments.}
Assume line segments \(U_1 \eqdef p_1^1p_1^2\) and \(U_n \eqdef p_n^1p_n^2\) cross; call their
intersection point \(s\).
Then we can use \cref{cor:line_segments} separately on pairs of
\(\{p_1^1s, sp_1^2\} \times \{p_n^1s, sp_n^2\}\).
Clearly, together this will cover the entire set of realisations of \(pq\) with \(p \in U_1\),
\(q \in U_n\), thus completing the checks.
\begin{lemma}\label{lem:line_main}
Given \(n \in \N^{> 0}\), \(n \geq 3\), for any imprecise curve modelled with line segments
\(\U = \langle U_1, \dots, U_n\rangle\) with \(U_i = p_i^1p_i^2 \subset \Rt\) for all \(i \in [n]\),
given a threshold \(\varepsilon \in \R^{> 0}\), we can check for both \(\delta \coloneqq \hs\) and
\(\delta \coloneqq \fr\), using procedures above, that
\[\max_{\pi \Subset \U} \delta(\pi, \langle\pi(1), \pi(n)\rangle) \leq \varepsilon\,.\]
\end{lemma}

\paragraph*{Convex polygons.}
Convex polygons whose interiors intersect can be partitioned along the intersection lines, so into
a convex polygon \(R \coloneqq U_1 \cap U_n\) and two sets of polygons
\(\Pm_1 \coloneqq \{P_1^1, \dots, P_1^k\}\) and
\(\Pm_n \coloneqq \{P_n^1, \dots, P_n^\ell\}\) for some \(k, \ell \in \N^{> 0}\).
Just as for line segments, we can look at pairs from \(\Pm_1 \times \Pm_n\) separately.
The pairs where \(R\) is involved are treated later.
Consider some \((P, Q) \in \Pm_1 \times \Pm_n\).
Note that \(P\) and \(Q\) are convex polygons with a convex cut-out, so the boundary forms a convex
chain, followed by a concave chain.
We need to compute some convex polygons \(P'\) and \(Q'\) with non-intersecting interiors that are
equivalent to \(P\) and \(Q\), so that we can apply the approaches from \cref{sec:non_int_pccs}.

We claim that we can simply take the convex hull of \(P\) and \(Q\) to obtain \(P'\) and \(Q'\).
Clearly, the resulting polygons will be convex.
Also, the concave chains of \(P\) are bounded by points \(s\) and \(t\) and are replaced with the
line segment \(st\); same happens for \(Q\) with point \(u\) and \(v\).
The points \(s, t, u, v\) are points of intersection of original polygons \(U_1\) and \(U_n\), so
they lie on the boundary of \(R\), and their order along that boundary can only be \(s, t, u, v\)
or \(s, t, v, u\).
Thus, it cannot happen that \(st\) crosses \(uv\), and it cannot be that \(uv\) is in the interior
of the convex hull of \(P\), as otherwise \(R\) would not be convex.
Hence, the interiors of \(P'\) and \(Q'\) cannot intersect, so they satisfy the necessary
conditions.

Finally, we need to show that the solution for \((P', Q')\) is equivalent to that for \((P, Q)\).
One direction is trivial, as \(P \subseteq P'\) and \(Q \subseteq Q'\); for the other
direction, consider any line segment that leaves \(P\) through the concave chain.
In our approach, we test the lines starting in \(s\) and \(t\); the established alignments
are connected into paths.
The paths \(A_i\) do not cross \(st\).
So, any alignment in the region of
\(\mathrm{CH}(P \cup Q) \setminus (P \cup Q)\) can also be made in the region
\(\mathrm{CH}(P' \cup Q') \setminus (P' \cup Q')\).
So, this approach yields valid solutions for all pairs not involving \(R\).

Now consider the pair \((R, R)\).
A curve may now consist of a single point, so the approach for the Fr\'echet and the Hausdorff
distance is the same: all the points of \(U_i\) need to be close enough to all the points of \(R\).
To check that, observe that the pair of points \(p \in U_i\) and \(q \in R\) that has maximal
distance has the property that \(p\) is an extreme point of \(U_i\) in direction \(qp\) and \(q\) is
an extreme point of \(R\) in direction \(pq\).
So, it suffices, starting at the rightmost point of \(U_i\) and leftmost point of \(R\) in some
coordinate system, to then rotate clockwise around both regions keeping track of the distance
between tangent points.
Note that only vertices need to be considered, as the extremal point cannot lie on an edge.
Finally, any other pair that involves \(R\) is covered by the stronger case of \((R, R)\): for any
line we can align every intermediate object to any point in \(R\).
\begin{lemma}\label{lem:polygon_main}
Given \(n \in \N^{> 0}\), \(n \geq 3\), for any imprecise curve modelled with convex polygons
\(\U = \langle U_1, \dots, U_n\rangle\) with \(U_i \subset \Rt\) for all \(i \in [n]\) and
\(V(U_i) = \{p_i^1, \dots, p_i^k\}\) for all \(i \in [n]\), \(k \in \N^{> 0}\), given a threshold
\(\varepsilon \in \R^{> 0}\), we can check for both \(\delta \coloneqq \hs\) and
\(\delta \coloneqq \fr\), using procedures above, that
\[\max_{\pi \Subset \U} \delta(\pi, \langle\pi(1), \pi(n)\rangle) \leq \varepsilon\,.\]
\end{lemma}

\section{Combining Steps}\label{sec:graph}
In \cref{sec:shortcut,sec:intermediate} we have established correctness of the routines that can be
used to check if a shortcut is valid under either the Hausdorff distance or the Fr\'echet distance.
In this \lcnamecref{sec:graph}, we summarise the approach, discuss the
shortcut graph, and analyse the running times.

\begin{lemma}\label{lem:shortcut}
Given \(n \in \N^{> 0}\), for any uncertain curve modelled with
indecisive points, disks, or PCCSs \(\U = \langle U_1, \dots, U_n\rangle\), and given a threshold
\(\varepsilon \in \R^{> 0}\), and fixing either \(\delta \coloneqq \hs\) or \(\delta \coloneqq \fr\),
if we can check in time \(T\) for any pair \(i, j \in [n]\), \(i < j\) that
\[\max_{\pi \Subset \U[i: j]} \delta(\pi, \langle\pi(1), \pi(j - i + 1)\rangle) \leq \varepsilon\,,\]
then in time \(\bO(Tn^2)\) we can find the shortest index subsequence \(I \subseteq [n]\) with
\(\lvert I\rvert = \ell\) such that for all \(j \in [\ell]\),
\[\max_{\pi \Subset \U[I(j): I(j + 1)]} \delta(\pi, \langle\pi(1), \pi(I(j + 1) - I(j) + 1)\rangle)
\leq \varepsilon\,.\]
\end{lemma}
\begin{proof}
The approach is simple: construct a graph \(G \coloneqq (V, E)\) with
\(V \coloneqq \{v_1, \dots, v_n\}\) and \((v_i, v_j) \in E\) if and only if
\(\max_{\pi \Subset \U[i: j]} \delta(\pi, \langle\pi(1), \pi(j - i + 1)\rangle) \leq \varepsilon\).
Clearly, this takes \(\bO(Tn^2)\) time.
Any path in the graph from \(v_1\) to \(v_n\) gives a subsequence for which the condition in the
statement of the \lcnamecref{lem:shortcut} holds; there are no simplifications that would not
correspond to such a path; thus, finding the shortest path in \(G\) using e.g.\@ BFS in time
\(\bO(n^2)\) indeed yields the answer.
\end{proof}
It is easy to see that the result of the \lcnamecref{lem:shortcut} is exactly the problem we were
trying to solve: obtaining a single simplification such that no matter which realisation of the
curve is chosen, the resulting realisation of the simplification is valid.

We now proceed to recap the methods for checking the shortcuts.
For indecisive points, one can test all combinations for the first and the last point of the
shortcut, as in \cref{lem:indecisive_main}, and for each such combination do the testing either for
the Hausdorff or the Fr\'echet distance, as in \cref{lem:hausdorff_ind,lem:frechet_indecisive}.

For imprecise points modelled with disks, it suffices to test the outer tangents and check some
extra conditions on the intermediate disks, as in
\cref{lem:hausdorff_disks_main,lem:frechet_disks_main}.
For the outer tangents, the testing can be done using the approaches of
\cref{lem:hausdorff_disks,lem:frechet_disks}.

For imprecise points modelled with line segments, one can split the first and the last one into
regions if they cross, as in \cref{lem:line_main}, and apply \cref{cor:line_segments} to each pair.
The testing of the outer tangents can be done using \cref{lem:hausdorff_pccs,lem:frechet_pccs} for
the Hausdorff and the Fr\'echet distance, respectively.

Finally, for imprecise points modelled with convex polygons, we again split the first and the last
one into regions if their interiors intersect, as in \cref{lem:polygon_main}, and apply
\cref{lem:hausdorff_pccs_main,lem:frechet_pccs_main}.
To test each shortcut with the fixed endpoints, we can again use
\cref{lem:hausdorff_pccs,lem:frechet_pccs}.

Having constructed the graph, we can find the shortest path through it from vertex corresponding to
\(U_1\) to that corresponding to \(U_n\), as discussed in \cref{lem:shortcut}.

\begin{theorem}\label{thm:main}
We can solve the problem of finding the shortest vertex-constrained simplification of an
uncertain curve, such that for any realisation the simplification is valid, both for the
Hausdorff and the Fr\'echet distance, and for uncertainty modelled using indecisive points,
disks, line segments, or convex polygons in time shown in \cref{tab:runningtime}.
\end{theorem}
\begin{proof}
Correctness of the approaches has been shown before.
For the running time, observe that we need \(\bO(n^2T)\) time in any setting, due to the shortcut
graph construction.

For indecisive points, when testing a shortcut we do \(\bO(nk)\)-time testing for \(\bO(k^2)\)
combinations of starting and ending points, where \(k\) is the number of options per point.

For disks, we do a linear number of constant-time checks and two linear-time checks, getting
\(T \in \bO(n)\).

For line segments, we also do two (three) linear-time checks per part; two line segments can be
split into at most two parts each, so we repeat the process four times.
Either way, we get \(T \in \bO(n)\).

Finally, for convex polygons, assume the complexity of each polygon is at most \(k\).
Assume the partitioning resulting from two intersecting polygons yields \(\ell_1\) and \(\ell_2\)
parts for the first and the second polygon, respectively.
Denote the two polygons \(P\) and \(Q\) and the resulting parts with \(P_1, \dots, P_{\ell_1}\) and
\(Q_1, \dots, Q_{\ell_2}\), respectively.
Suppose part \(P_i\) has complexity \(k_i\) and part \(Q_j\) has complexity \(k'_j\), so
\(\lvert V(P_i)\rvert = k_i\) and \(\lvert V(Q_j)\rvert = k'_j\) for some \(i \in [\ell_1]\),
\(j \in [\ell_2]\).
We know that every vertex of the original polygons occurs in a constant number of parts, so
\(\sum_{i = 1}^{\ell_1} k_i \in \bO(k)\) and \(\sum_{j = 1}^{\ell_2} k'_j \in \bO(k)\); we also know
\(\ell_1 + \ell_2 \in \bO(k)\).
We consider all pairs from \(P\) and \(Q\), and for each pair we triangulate and do the checks on
the triangulation.
The triangulation can be done in time \(\bO((k_i + k'_j) \cdot \log (k_i + k'_j))\), yielding
\(\bO(k_i + k'_j)\) lines, each of which is tested in time \(\bO(nk)\).
The testing dominates, so we need \(\bO((k_i + k'_j) \cdot nk)\) time.
We are interested in
\[\sum_{i = 1}^{\ell_1} \sum_{j = 1}^{\ell_2} \bO((k_i + k'_j) \cdot nk) =
\bO(nk) \cdot \sum_{i = 1}^{\ell_1} \sum_{j = 1}^{\ell_2} \bO(k_i + k'_j) = \bO(nk^3)\,.\]
So, \(T \in \bO(nk^3)\) both for the Fr\'echet and the Hausdorff distance.
\end{proof}

\bibliographystyle{plainurl}
\bibliography{references}

\begin{thebibliography}{10}

\bibitem{agarwal:2016}
Pankaj~K. Agarwal, Boris Aronov, Sariel Har-Peled, Jeff~M. Phillips, Ke~Yi, and
  Wuzhou Zhang.
\newblock Nearest-neighbor searching under uncertainty {II}.
\newblock {\em {ACM} Transactions on Algorithms ({TALG})}, 13(1):3:1--3:25,
  December 2016.
\newblock \href {https://doi.org/10.1145/2955098} {\path{doi:10.1145/2955098}}.

\bibitem{agarwal:2017}
Pankaj~K. Agarwal, Alon Efrat, Swaminathan Sankararaman, and Wuzhou Zhang.
\newblock Nearest-neighbor searching under uncertainty {I}.
\newblock {\em Discrete \& Computational Geometry}, 58(3):705--745, July 2017.
\newblock \href {https://doi.org/10.1007/s00454-017-9903-x}
  {\path{doi:10.1007/s00454-017-9903-x}}.

\bibitem{agarwal}
Pankaj~K. Agarwal, Sariel Har-Peled, Nabil~H. Mustafa, and Yusu Wang.
\newblock Near-linear time approximation algorithms for curve simplification.
\newblock {\em Algorithmica}, 42(3):203--219, July 2005.
\newblock \href {https://doi.org/10.1007/s00453-005-1165-y}
  {\path{doi:10.1007/s00453-005-1165-y}}.

\bibitem{agarwal:2000}
Pankaj~K. Agarwal and Kasturi~R. Varadarajan.
\newblock Efficient algorithms for approximating polygonal chains.
\newblock {\em Discrete \& Computational Geometry}, 23(2):273--291, 2000.
\newblock \href {https://doi.org/10.1007/PL00009500}
  {\path{doi:10.1007/PL00009500}}.

\bibitem{ahn_imprecise}
Hee-Kap Ahn, Christian Knauer, Marc Scherfenberg, Lena Schlipf, and Antoine
  Vigneron.
\newblock Computing the discrete {F}r\'{e}chet distance with imprecise input.
\newblock {\em International Journal of Computational Geometry \&
  Applications}, 22(01):27--44, 2012.
\newblock \href {https://doi.org/10.1142/S0218195912600023}
  {\path{doi:10.1142/S0218195912600023}}.

\bibitem{bbmm_segment}
Sander P.~A. Alewijnse, Kevin Buchin, Maike Buchin, Stef Sijben, and Michel~A.
  Westenberg.
\newblock Model-based segmentation and classification of trajectories.
\newblock {\em Algorithmica}, 80(8):2422--2452, August 2018.
\newblock \href {https://doi.org/10.1007/s00453-017-0329-x}
  {\path{doi:10.1007/s00453-017-0329-x}}.

\bibitem{barequet}
Gill Barequet, Danny~Z. Chen, Ovidiu Daescu, Michael~T. Goodrich, and Jack~S.
  Snoeyink.
\newblock Efficiently approximating polygonal paths in three and higher
  dimensions.
\newblock {\em Algorithmica}, 33(2):150--167, 2002.
\newblock \href {https://doi.org/10.1007/s00453-001-0096-5}
  {\path{doi:10.1007/s00453-001-0096-5}}.

\bibitem{bringmann_simpl}
Karl Bringmann and Bhaskar~Ray Chaudhury.
\newblock Polyline simplification has cubic complexity.
\newblock In {\em 35th International Symposium on Computational Geometry
  ({SoCG} 2019)}, volume 129 of {\em Leibniz International Proceedings in
  Informatics ({LIPIcs})}, pages 18:1--18:16, Dagstuhl, Germany, 2019. Schloss
  Dagstuhl~--~Leibniz-Zentrum f\"{u}r Informatik.
\newblock \href {https://doi.org/10.4230/LIPIcs.SoCG.2019.18}
  {\path{doi:10.4230/LIPIcs.SoCG.2019.18}}.

\bibitem{uncertaincurves}
Kevin Buchin, Chenglin Fan, Maarten L\"{o}ffler, Aleksandr Popov, Benjamin
  Raichel, and Marcel Roeloffzen.
\newblock {F}r\'{e}chet distance for uncertain curves.
\newblock In {\em 47th International Colloquium on Automata, Languages, and
  Programming ({ICALP} 2020)}, volume 168 of {\em LIPIcs}, pages 20:1--20:20,
  Dagstuhl, Germany, 2020. Schloss Dagstuhl--Leibniz-Zentrum f{\"u}r
  Informatik.
\newblock \href {https://doi.org/10.4230/LIPIcs.ICALP.2020.20}
  {\path{doi:10.4230/LIPIcs.ICALP.2020.20}}.

\bibitem{buchin_progressive}
Kevin Buchin, Maximilian Konzack, and Wim Reddingius.
\newblock Progressive simplification of polygonal curves.
\newblock {\em Computational Geometry}, 88:101620:1--101620:18, 2020.
\newblock \href {https://doi.org/10.1016/j.comgeo.2020.101620}
  {\path{doi:10.1016/j.comgeo.2020.101620}}.

\bibitem{buchin:2011}
Kevin Buchin, Maarten L\"{o}ffler, Pat Morin, and Wolfgang Mulzer.
\newblock Preprocessing imprecise points for {D}elaunay triangulation:
  Simplified and extended.
\newblock {\em Algorithmica}, 61(3):674--693, November 2011.
\newblock \href {https://doi.org/10.1007/s00453-010-9430-0}
  {\path{doi:10.1007/s00453-010-9430-0}}.

\bibitem{bbmm_patterns}
Kevin Buchin, Stef Sijben, T.~Jean~Marie Arseneau, and Erik~P. Willems.
\newblock Detecting movement patterns using {B}rownian bridges.
\newblock In {\em Proceedings of the 20th International Conference on Advances
  in Geographic Information Systems}, {SIGSPATIAL} '12, pages 119--128, New
  York, NY, USA, 2012. ACM.
\newblock \href {https://doi.org/10.1145/2424321.2424338}
  {\path{doi:10.1145/2424321.2424338}}.

\bibitem{sijben_dfd}
Maike Buchin and Stef Sijben.
\newblock Discrete {F}r\'{e}chet distance for uncertain points, 2016.
\newblock Presented at EuroCG 2016, Lugano, Switzerland.
\newblock URL:
  \url{http://www.eurocg2016.usi.ch/sites/default/files/paper_72.pdf} [cited
  2019-07-10].

\bibitem{visibility}
Leizhen Cai and Mark Keil.
\newblock Computing visibility information in an inaccurate simple polygon.
\newblock {\em International Journal of Computational Geometry \&
  Applications}, 7:515--538, December 1997.
\newblock \href {https://doi.org/10.1142/S0218195997000326}
  {\path{doi:10.1142/S0218195997000326}}.

\bibitem{chan_chin}
W.~S. Chan and Francis Chin.
\newblock Approximation of polygonal curves with minimum number of line
  segments or minimum error.
\newblock {\em International Journal of Computational Geometry and
  Applications}, 6(1):59--77, 1996.
\newblock \href {https://doi.org/10.1142/S0218195996000058}
  {\path{doi:10.1142/S0218195996000058}}.

\bibitem{douglas_peucker}
David~H. Douglas and Thomas~K. Peucker.
\newblock Algorithms for the reduction of the number of points required to
  represent a digitized line or its caricature.
\newblock {\em Cartographica: The International Journal for Geographic
  Information and Geovisualization}, 10(2):112--122, 1973.
\newblock \href {https://doi.org/10.3138/FM57-6770-U75U-7727}
  {\path{doi:10.3138/FM57-6770-U75U-7727}}.

\bibitem{driemel:2013}
Anne Driemel, Herman Haverkort, Maarten L\"{o}ffler, and Rodrigo~I. Silveira.
\newblock Flow computations on imprecise terrains.
\newblock {\em Journal of Computational Geometry ({JoCG})}, 4(1):38--78, 2013.
\newblock \href {https://doi.org/10.20382/jocg.v4i1a3}
  {\path{doi:10.20382/jocg.v4i1a3}}.

\bibitem{minimising_ply}
William Evans, David Kirkpatrick, Maarten L\"{o}ffler, and Frank Staals.
\newblock Competitive query strategies for minimising the ply of the potential
  locations of moving points.
\newblock In {\em Proceedings of the Twenty-Ninth Annual Symposium on
  Computational Geometry}, {SoCG} '13, pages 155--164, New York, NY, USA, 2013.
  ACM.
\newblock \href {https://doi.org/10.1145/2462356.2462395}
  {\path{doi:10.1145/2462356.2462395}}.

\bibitem{gray:2012}
Chris Gray, Frank Kammer, Maarten L\"{o}ffler, and Rodrigo~I. Silveira.
\newblock Removing local extrema from imprecise terrains.
\newblock {\em Computational Geometry}, 45(7):334--349, 2012.
\newblock \href {https://doi.org/10.1016/j.comgeo.2012.02.002}
  {\path{doi:10.1016/j.comgeo.2012.02.002}}.

\bibitem{gudmundsson}
Joachim Gudmundsson, Jyrki Katajainen, Damian Merrick, Cahya Ong, and Thomas
  Wolle.
\newblock Compressing spatio-temporal trajectories.
\newblock {\em Computational Geometry}, 42(9):825--841, November 2009.
\newblock \href {https://doi.org/10.1016/j.comgeo.2009.02.002}
  {\path{doi:10.1016/j.comgeo.2009.02.002}}.

\bibitem{guibas}
Leonidas~J. Guibas, John~E. Hershberger, Joseph S.~B. Mitchell, and Jack~S.
  Snoeyink.
\newblock Approximating polygons and subdivisions with minimum-link paths.
\newblock {\em International Journal of Computational Geometry \&
  Applications}, 3(4):383--415, 1993.
\newblock \href {https://doi.org/10.1142/S0218195993000257}
  {\path{doi:10.1142/S0218195993000257}}.

\bibitem{imai_iri}
Hiroshi Imai and Masao Iri.
\newblock Computational-geometric methods for polygonal approximations of a
  curve.
\newblock {\em Computer Vision, Graphics, and Image Processing}, 36(1):31--41,
  1986.
\newblock \href {https://doi.org/10.1016/S0734-189X(86)80027-5}
  {\path{doi:10.1016/S0734-189X(86)80027-5}}.

\bibitem{jorgensen:2011}
Allan J{\o}rgensen, Jeff~M. Phillips, and Maarten L\"{o}ffler.
\newblock Geometric computations on indecisive points.
\newblock In {\em Algorithms and Data Structures ({WADS} 2011)}, volume 6844 of
  {\em Lecture Notes in Computer Science}, pages 536--547, Berlin, Germany,
  2011. Springer Berlin Heidelberg.
\newblock \href {https://doi.org/10.1007/978-3-642-22300-6_45}
  {\path{doi:10.1007/978-3-642-22300-6_45}}.

\bibitem{knauer_hausdorff}
Christian Knauer, Maarten L\"{o}ffler, Marc Scherfenberg, and Thomas Wolle.
\newblock The directed {H}ausdorff distance between imprecise point sets.
\newblock {\em Theoretical Computer Science}, 412(32):4173--4186, 2011.
\newblock \href {https://doi.org/10.1016/j.tcs.2011.01.039}
  {\path{doi:10.1016/j.tcs.2011.01.039}}.

\bibitem{loeffler}
Maarten L\"{o}ffler.
\newblock {\em Data Imprecision in Computational Geometry}.
\newblock PhD thesis, Universiteit Utrecht, October 2009.
\newblock URL:
  \url{https://dspace.library.uu.nl/bitstream/handle/1874/36022/loffler.pdf}
  [cited 2019-06-15].

\bibitem{loeffler:2014}
Maarten L\"{o}ffler and Wolfgang Mulzer.
\newblock Unions of onions: Preprocessing imprecise points for fast onion
  decomposition.
\newblock {\em Journal of Computational Geometry ({JoCG})}, 5(1):1--13, 2014.
\newblock \href {https://doi.org/10.20382/jocg.v5i1a1}
  {\path{doi:10.20382/jocg.v5i1a1}}.

\bibitem{prob_loeffler}
Maarten L\"{o}ffler and Jeff~M. Phillips.
\newblock Shape fitting on point sets with probability distributions: {ESA}
  2009.
\newblock In {\em Algorithms}, number 5757 in LNCS, pages 313--324, Berlin,
  Germany, 2009. Springer Berlin Heidelberg.
\newblock \href {http://arxiv.org/abs/0812.2967v1} {\path{arXiv:0812.2967v1}},
  \href {https://doi.org/10.1007/978-3-642-04128-0_29}
  {\path{doi:10.1007/978-3-642-04128-0_29}}.

\bibitem{loeffler:2010}
Maarten L\"{o}ffler and Jack~S. Snoeyink.
\newblock {D}elaunay triangulations of imprecise points in linear time after
  preprocessing.
\newblock {\em Computational Geometry: Theory and Applications},
  43(3):234--242, 2010.
\newblock \href {https://doi.org/10.1016/j.comgeo.2008.12.007}
  {\path{doi:10.1016/j.comgeo.2008.12.007}}.

\bibitem{loeffler:2006}
Maarten L\"{o}ffler and Marc van Kreveld.
\newblock Largest and smallest tours and convex hulls for imprecise points.
\newblock In {\em Algorithm Theory -- {SWAT} 2006}, volume 4059 of {\em Lecture
  Notes in Computer Science}, pages 375--387, Berlin, Germany, 2006. Springer
  Berlin Heidelberg.
\newblock \href {https://doi.org/10.1007/11785293_35}
  {\path{doi:10.1007/11785293_35}}.

\bibitem{melkman}
Avraham Melkman and Joseph O'Rourke.
\newblock On polygonal chain approximation.
\newblock In Godfried~T. Toussaint, editor, {\em Computational Morphology},
  volume~6 of {\em Machine Intelligence and Pattern Recognition}, pages 87--95.
  Elsevier Science Publishers, 1988.
\newblock \href {https://doi.org/10.1016/B978-0-444-70467-2.50012-6}
  {\path{doi:10.1016/B978-0-444-70467-2.50012-6}}.

\bibitem{pei:2007}
Jian Pei, Bin Jiang, Xuemin Lin, and Yidong Yuan.
\newblock Probabilistic skylines on uncertain data.
\newblock In {\em Proceedings of the 33rd International Conference on Very
  Large Data Bases}, pages 15--26. VLDB Endowment, September 2007.
\newblock \href {https://doi.org/10.5555/1325851.1325858}
  {\path{doi:10.5555/1325851.1325858}}.

\bibitem{popov}
Aleksandr Popov.
\newblock Similarity of uncertain trajectories.
\newblock Master's thesis, Eindhoven University of Technology, November 2019.
\newblock URL:
  \url{https://research.tue.nl/en/studentTheses/similarity-of-uncertain-trajectories}
  [cited 2019-12-18].

\bibitem{ramer}
Urs Ramer.
\newblock An iterative procedure for the polygonal approximation of plane
  curves.
\newblock {\em Computer Graphics and Image Processing}, 1(3):244--256, 1972.
\newblock \href {https://doi.org/10.1016/S0146-664X(72)80017-0}
  {\path{doi:10.1016/S0146-664X(72)80017-0}}.

\bibitem{suri:2013}
Subhash Suri, Kevin Verbeek, and Hakan Y{\i}ld{\i}z.
\newblock On the most likely convex hull of uncertain points.
\newblock In {\em Algorithms -- {ESA} 2013}, volume 8125 of {\em Lecture Notes
  in Computer Science}, pages 791--802, Berlin, Germany, 2013. Springer Berlin
  Heidelberg.
\newblock \href {https://doi.org/10.1007/978-3-642-40450-4_67}
  {\path{doi:10.1007/978-3-642-40450-4_67}}.

\bibitem{kerkhof}
Mees van~de Kerkhof, Irina Kostitsyna, Maarten L\"{o}ffler, Majid Mirzanezhad,
  and Carola Wenk.
\newblock Global curve simplification.
\newblock In {\em 27th Annual European Symposium on Algorithms ({ESA} 2019)},
  volume 144 of {\em Leibniz International Proceedings in Informatics
  ({LIPIcs})}, pages 67:1--67:14, Dagstuhl, Germany, 2019. Schloss Dagstuhl --
  Leibniz-Zentrum f\"{u}r Informatik.
\newblock \href {https://doi.org/10.4230/LIPIcs.ESA.2019.67}
  {\path{doi:10.4230/LIPIcs.ESA.2019.67}}.

\bibitem{kreveld:2010}
Marc van Kreveld, Maarten L\"{o}ffler, and Joseph S.~B. Mitchell.
\newblock Preprocessing imprecise points and splitting triangulations.
\newblock {\em {SIAM} Journal on Computing}, 39(7):2990--3000, May 2010.
\newblock \href {https://doi.org/10.1137/090753620}
  {\path{doi:10.1137/090753620}}.

\end{thebibliography}
\end{document}